\documentclass[peerreview,draftcls,twocolumn,12pt,a4paper]{IEEEtran}
\usepackage[T1]{fontenc}
\usepackage[active]{srcltx} % SRC Specials for DVI Searching
\usepackage{amsmath}
\usepackage{color}
\usepackage{amssymb}
\usepackage{color}
\usepackage[final]{graphicx}
\usepackage{amsbsy}
\usepackage{url}
\usepackage{enumerate}
\usepackage{epsfig}
\usepackage{latexsym}
\usepackage{graphicx}
\usepackage{euscript}
\usepackage{mathrsfs}
\usepackage{psfrag}
\usepackage{cite}
\usepackage{url}
\usepackage[english]{babel}
\newcommand{\Tcp}{\mathbf{T}_{\text{cp}}}

\newcommand{\IeeeTIT}{{\em IEEE Trans.\ Inf.\ Theory\/}}
\newcommand{\IeeeTSP}{{\em IEEE Trans.\ Signal Process.\/}}

\newcommand{\IeeeSPLETT}{{\em IEEE Signal Process.\ Lett.\/}}
\newcommand{\IeeeTWC}{{\em IEEE Trans.\ Wireless Commun.\/}}
\newcommand{\IeeeJSAC}{{\em IEEE J.\ Select.\ Areas Commun.\/}}
\newcommand{\IeeeTVT}{{\em IEEE Trans.\ Veh.\ Technol.\/}}

\newcommand{\IeeeSPMAG}{{\em IEEE Signal Processing Magazine\/}}

\newtheorem{lemma}{Lemma}

\newtheorem{definition}{Definition}

\def\a{{\mathbf a}}
\def\J{{\mathbf J}}

\def\A{{\mathbf A}}
\def\B{{\mathbf B}}
\def\R{{\mathbf R}}
\def\F{{\mathbf F}}
\def\G{{\mathbf G}}

\def\Sf{{\mathbf S}_{\text{f}}}
\def\Sb{{\mathbf S}_{\text{b}}}
\def\bOmega{{\bm \Omega}}

\def\H{{\mathbf H}}
\def\Acal{\bm{\cal A}}
\def\Acalover{{\bm{\cal B}}}
\def\Ccal{{\bm{\cal C}}}
\def\Hcal{\bm{\cal H}}
\def\Gcal{\bm{\cal G}}
\def\Fcal{\bm{\cal F}}
\def\Dcal{\bm{\cal D}}
\def\V{{\mathbf V}}

\def\X{{\mathbf X}}
\def\x{{\mathbf x}}
\def\h{{\mathbf h}}
\def\y{{\mathbf y}}
\def\z{{\mathbf z}}

\def\r{{\mathbf r}}
\def\v{{\mathbf v}}

\def\u{{\mathbf u}}

\def\f{{\mathbf f}}
\def\xsuuc{\x_{\SU,\text{I}}}
\def\xsuvc{\x_{\SU,\text{II}}}
\def\expect{{\mathsf{E}}}
\def\PU{\text{PU}}
\def\SU{\text{SU}}

\def\Prob{\mathsf{P}}
\def\PPU{\mathsf{P}_{\PU}}
\def\PSU{\mathsf{P}_{\SU}}
\def\SNR{\mathsf{SNR}}
\def\ASNR{\mathsf{ASNR}}

\newcommand{\Mvc}{{M_{\text{vc}}}}
\newcommand{\Ob}{\mathbf{O}}
\newcommand{\ob}{\mathbf{0}}
\newcommand{\IPUvc}{{\EuScript{I}_{\text{PU},\text{vc}}}}
\newcommand{\IPUuc}{{\EuScript{I}_{\text{PU},\text{uc}}}}

\newcommand{\Lcppu}{{L_{\text{cp}}}}
\newcommand{\Lpre}{{L_\text{SU}}}
\newcommand{\eqdef}{\triangleq}
\newcommand{\bm}[1]{{\boldsymbol{#1}}}

\newcommand{\setC}{\mathbb{C}}
\newcommand{\setR}{\mathbb{R}}
\newcommand{\setZ}{\mathbb{Z}}

\newcommand{\setF}{\mathbb{F}}

\newcommand{\Widft}{\mathbf{W}_{\text{IDFT}}}
\newcommand{\Wdft}{\mathbf{W}_{\text{DFT}}}
\newcommand{\Widftdown}{\overline{\mathbf{W}}_{\text{IDFT}}}
\newcommand{\I}{\mathbf{I}}
\newcommand{\C}{\mathbf{C}}
\newcommand{\diag}{{\text{diag}}}

\newcommand{\herm}{\text{H}}
\newcommand{\trasp}{\text{T}}

\newcommand{\rank}{{\text{rank}}}
\newcommand{\Rvvpu}{\mathbf{R}_{\v_{\PU}}}
\newcommand{\Rvvsu}{\mathbf{R}_{\v_{\SU}}}

\newcommand{\Fxsu}{\Fcal}

\def\CPUd{\mathsf{C}_{\PU, \text{direct}}}

\def\ASNRd{\ASNR_{13, \text{direct}}}
\def\ASNRdd{\ASNR_{24, \text{direct}}}
\def\ASNRddd{\ASNR_{14, \text{direct}}}

\def\MI{\mathsf{I}}
\def\Prob{\mathsf{Prob}}
\def\CPU{\mathsf{C}_{\PU}}
\def\CPUworst{\mathsf{C}_{\PU,\text{worst}}}

\def\CPUlower{\mathsf{C}_{\PU,\text{lower}}}
\def\PPUout{\mathsf{Prob}_{\PU,\text{out},m}}
\def\CSU{\mathsf{C}_{\SU}}
\def\CSUlower{\mathsf{C}_{\SU, \text{lower}}}
\def\CSUloweropt{\mathsf{C}_{\SU, \text{lower},\text{CSIT}}}
\def\CSUlowerunif{\mathsf{C}_{\SU, \text{lower},\text{NOCSIT}}}

\def\C124{\mathsf{C}_{\PU \to \SU}}

\def\ASNRP{\Gamma_{3, m}}
\def\ASNRS{\Gamma_{4, m}}

\def\bdm#1\edm{\begin{displaymath}#1\end{displaymath}}
\def\be#1\ee{\begin{equation}#1\end{equation}}
\def\barr#1\earr{\begin{align}#1\end{align}}

\newcommand{\sfootnote}[1]{{\footnote{\setlength{\baselineskip}{4.0mm}#1}}}

\begin{document}
\title{
\setlength{\baselineskip}{10.0mm}
Convolutive superposition for multicarrier cognitive radio systems}

%\author{Donatella Darsena,~\IEEEmembership{Member,~IEEE,} 
%         Giacinto Gelli, and \\
%        Francesco Verde,~\IEEEmembership{Senior Member,~IEEE}
%
\author{Donatella Darsena, Giacinto Gelli, and Francesco Verde
\thanks{
\setlength{\baselineskip}{4.0mm}
D.~Darsena is with the Department of Engineering,
Parthenope University, Naples I-80143, Italy (e-mail: darsena@uniparthenope.it).
G.~Gelli and F.~Verde are with the
Department of Electrical Engineering and
Information Technology, University Federico II, Naples I-80125,
Italy [e-mail: (gelli,f.verde)@unina.it].}
}

\maketitle
\vspace{-10mm}
\begin{abstract}

\setlength{\baselineskip}{5.0mm}
Recently, we proposed a spectrum-sharing
paradigm for single-carrier cognitive radio (CR) networks,
where a secondary user (SU) is able to maintain or
even improve the performance of
a primary user (PU) transmission, while also obtaining
a low-data rate channel for its own communication.
According to such a scheme, a simple multiplication 
is used to superimpose one SU symbol
on a block of multiple PU symbols.
The scope of this paper is to extend
such a paradigm to a multicarrier
CR network, where the PU employs an
orthogonal frequency-division multiplexing
(OFDM) modulation scheme.
To improve
its achievable data rate,
besides transmitting over the subcarriers unused by the PU,
the SU is also allowed to transmit multiple block-precoded symbols in parallel
over the OFDM subcarriers used by the primary system. 
Specifically, the SU convolves  its block-precoded symbols with the
received PU data in the time-domain,  
which gives rise to the term {\em convolutive  superposition}.
An information-theoretic
analysis of the proposed scheme is developed, which
considers different amounts of network state information
at the secondary transmitter, as well as 
different precoding strategies for the SU.
Extensive simulations illustrate the merits of our
analysis and designs, in comparison with conventional CR schemes, 
by considering as performance indicators
the ergodic capacity of the considered systems.

\end{abstract}

\vspace{-5mm}
\begin{IEEEkeywords}
\vspace{-5mm}
Cognitive radio, channel capacity,
multicarrier modulation, superposition,
precoding design.
\end{IEEEkeywords}

\vspace{-5mm}

\section{Introduction}

Due to the explosive growth in wireless data services, 
mainly driven by video communications, 
next-generation wireless systems will require 
significant advances \cite{Andrews2014,Boccardi2014} in terms of
data-rate, latency, and energy consumptions, 
as well as improved networking
and resource allocation procedures. 
Moreover, according to the emerging 
``Internet of Things'' (IoT) paradigm 
and the diffusion of machine-to-machine communications, 
next-generation wireless systems must be able to support an 
enormous number of  low-rate devices,
which will require new approaches and policies for spectrum allocation and
management, including new forms of {\em spectrum sharing}, where
{\em cognitive radio} (CR) approaches \cite{Haykin2005,Biglieri_book} are 
expected to play a major role.
In CR techniques, secondary users (SUs) share 
a portion of the spectrum with licensed or unlicensed primary users (PUs).
Such an approach is beneficial, e.g., for ultradense wireless systems \cite{Wang2014}, where medium-to-low-rate SU terminals might share the spectrum 
with high-rate PU devices.

The cognitive radio approach stems from the fact that
a major part of the licensed and unlicensed  spectrum 
is typically unused for significant periods of time,
so called {\em spectrum holes} or {\em white spaces}.
Therefore, a simple opportunistic access paradigm consists
of allowing the SUs to transmit in an orthogonal fashion 
(space, time or frequency) relative to the PU signals,
which will be referred to as {\em orthogonal CR (OCR)}.
\cite{Haykin2005,Biglieri_book}.
Such an approach requires a possible multidimensional 
space-time-frequency detection of PU users,
called {\em spectrum sensing} \cite{Haykin2005,Biglieri_book}.  
However, accurate detection of a vacant spectrum is not an
easy  task \cite{Liang2008}. Moreover, next-generation wireless systems
mandate non-orthogonal primary and secondary transmissions \cite{Andrews2014},
which will be referred to as {\em non-orthogonal CR (NOCR)}.

There are two different visions in CR to accomplish spectrum sharing
on a non-orthogonal basis \cite{Haykin2005,Biglieri_book}: 
(i) SUs can share PU  communications resources, provided that they keep the interference to PU transmissions (so called {\em interference temperature} \cite{Haykin2005}) below a very low threshold;
(ii) sophisticated encoding and decoding techniques are used 
to remove all (or part of) the mutual interference between
PU and SU transmissions  \cite{Yu,Khisti, Devroye, Jovi},
in order to relax the threshold on the SU transmission powers.
In the former paradigm, one of the major problem is
to determine the interference level a secondary transmitter 
causes to a primary receiver; in the latter one, sophisticated 
encoding techniques like dirty paper coding (DPC) \cite{Costa}
require {\em a priori} knowledge of the primary
user's transmitted data and/or how this 
sequence is encoded (codebook). 
The underlying common feature of both approaches
is the {\em additive superposition} of PU and SU transmissions
(additive interference channel \cite{Devroye}), 
i.e, PU and SU signals add up.

Recently, we have proposed in \cite{Verde1,Verde2} a different 
NOCR paradigm where the arithmetic operation of multiplication 
is used to superimpose a single SU symbol on a primary signal  
composed of multiple PU symbols, through a single-channel
amplify-and-forward (AF) protocol. The main advantages of such a scheme
can be summarized as follows: (i) under non-restrictive
conditions \cite{Verde2}, the SU can transmit without a power constraint,
while keeping the desirable property of not  degrading (but even improving) 
the PU performance; (ii) {\em a priori} knowledge of the PU data 
at the SU is not required.
However, the main limitation of the single-channel approach in \cite{Verde1,Verde2} is that only low-data rates can  be achieved 
by the SU.

The aim of this paper is to improve the achievable data rates of the
SU by introducing the concept of 
{\em convolutive  superposition}, whereby the SU data are superimposed 
on the PU received signal by means of a time-domain convolution.  
Such a new code construction extends the 
multiplicative superposition scheme of \cite{Verde1,Verde2}
along three lines:\sfootnote{Preliminary results of such an 
extension are reported in \cite{Verde3}.} 

\begin{enumerate}[1)]

\itemsep=0mm

\item
We consider a CR system where modulation is based on 
{\em orthogonal frequency-division 
multiplexing (OFDM)} \cite{Stevenson2009}, due to its advantages in multipath resistance, performance, spectral efficiency, flexibility, and
computational complexity. The multicarrier nature of the PU transmission allows the SU to be active over multiple primary subchannels, thereby attaining larger transmission rates compared to the single-channel scheme considered in \cite{Verde1,Verde2}.

\item
For each PU subchannel, we allow the SU to transmit multiple symbols within a single OFDM symbol interval of the primary system, by jointly exploiting both white spaces (i.e., unused subcarriers of the PU signal) and dirty spaces (i.e., subcarriers used by the PU). 

\item
The multi-channel multi-symbol nature of
the secondary transmission introduces 
additional   degrees
of   freedom   with   respect   to   the single-channel approach
of \cite{Verde1,Verde2}: the distribution of the available power over the transmit
dimensions. In this regard, we develop {\em precoding} strategies for the 
SU transmission by considering either the case when channel state information (CSI) is available at the secondary transmitter or this knowledge is missing. 

\end{enumerate}
The theoretical performance analysis of the proposed scheme is
based on input-output mutual information and ergodic capacity of both
the primary and secondary systems.\sfootnote{The ergodic 
capacity serves as a useful upper
bound on the performance of any communication system
and it is to some extent amenable to analytic studies;
it can be achieved if the length of
the codebook is long enough to reflect the
ergodic nature of fading
(i.e.,  the transmission duration of the
codeword is much greater than the
channel coherence time) \cite{Biglieri}.
At rates lower than the ergodic capacity, there exist coding strategies ensuring that
the average bit error rate (BER) decays exponentially with the codebook length \cite{Gallager_book, CoverThomas_book}.}
Results of comparisons studies with other CR approaches 
are also reported in terms of ergodic capacity.

The paper is organized as follows.  The system model and the considered 
communication scheme are described in Section~\ref{sec:model}.
The capacity analysis for the PU is carried out in Section~\ref{sec:analysis-pu}.
The information-theoretic analysis and precoding designs for the SU
are developed in Section~\ref{sec:analysis-su}, by considering different amounts 
of CSI at both ends of the communication link.
Numerical results are reported in Section~\ref{sec:simul}, aimed at
corroborating our theoretical findings. Finally,  the main results obtained in the paper are summarized in Section~\ref{sec:concl}.

\begin{figure}[tbp]
\centering
\psfrag{n1\r}{$\,1$}
\psfrag{n2\r}{$\,2$}
\psfrag{n3\r}{$\,3$}
\psfrag{n4\r}{$\,4$}
\psfrag{h12\r}{$\H_{12}$}
\psfrag{h13\r}{$\H_{13}$}
\psfrag{h14\r}{$\H_{14}$}
\psfrag{h23\r}{$\H_{23}$}
\psfrag{h24\r}{$\H_{24}$}
\psfrag{PTx\r}{PTx}
\psfrag{PRx\r}{PRx}
\psfrag{STx\r}{STx}
\psfrag{SRx\r}{SRx}
\psfrag{angsu\r}{$\vartheta$}
\includegraphics[width=0.70\linewidth, trim = 0 0 0 0]{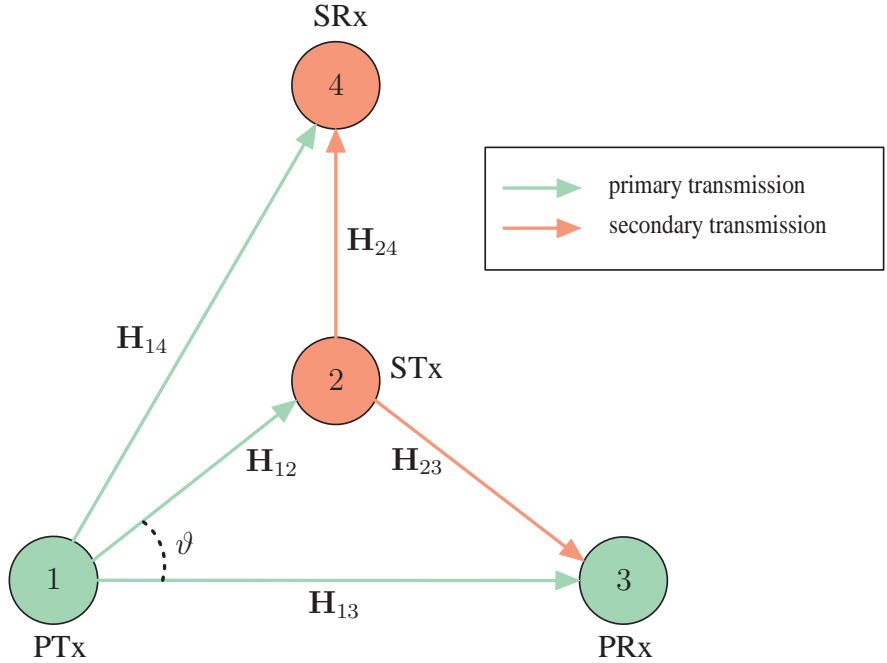}
\caption{
The wireless network model: in green,
the PU transmitter/receiver nodes, in red
the SU transmitter/receiver nodes.
}
\label{fig:figure_1}
\end{figure}

\section{The considered cognitive radio system model}
\label{sec:model}

We consider a multicarrier cognitive network (see Fig.~\ref{fig:figure_1}) 
composed  by a 
primary transmitter-receiver pair (nodes PTx and PRx) and one
secondary transmitter-receiver pair (nodes STx and SRx).
The nodes have a single antenna and 
operate in half-duplex mode, except for
the STx which is  equipped  with two antennas (one receive 
antenna and one transmit antenna) that enable a 
full-duplex operation.\sfootnote{We assume that
the transmit chain of the STx is adequately isolated 
from its receive chain \cite{Bliss}, such that self-interference 
is negligible in the receive chain circuitry.}
The PU employs OFDM
modulation with 
$M$ subcarriers, a cyclic 
prefix (CP) of length $\Lcppu < M$, and 
symbol period $T_{\PU} \eqdef P \, T_c$, where
$P \eqdef M + \Lcppu$ and $T_c$ denotes the 
sampling period of the PU system.
Only $Q$ out of the $M$ available subcarriers are utilized, 
whereas the remaining
$\Mvc \eqdef M-Q > 0$ ones are unmodulated and 
called {\em virtual subcarriers (VCs)}.
The STx exploits the PU transmission to deliver
to the SRx its own data, by simultaneously transmitting over
the same subcarriers of the PU, as described
in Subsection~\ref{sec:STx-trans}.
It is assumed that the SU perfectly knows the allocation of the VCs
within the PU frequency range.\sfootnote{Such a knowledge
can be available {\em a priori} or obtained, e.g., by means of spectrum sensing 
techniques \cite{Haykin2005}.}

\subsection{General assumptions regarding channels and noise}
\label{sec:assumpt}

During an interval of duration $T_{\PU}$,
the wireless channel between the pair of nodes 
$(i,j)$, for $i \in \{1,2\}$ and $i \neq j \in \{2,3,4\}$, 
 is modeled as a causal linear time-invariant (LTI) system
spanning at the most $L_{ij}>0$
sampling interval of the PU, i.e., its discrete-time impulse response
obeys $\widetilde{h}_{ij}( \ell) \equiv 0$ 
for $ \ell \not \in \{0,1,\ldots, L_{ij}\}$.
Such an impulse response
is fixed during the transmission of
one OFDM symbol, but is allowed to independently change
from one symbol to another.\sfootnote{For simplicity's sake, 
we will not explicitly indicate the dependence of the channels' parameters 
on the PU symbol period.}  
In the sequel, we will denote with $\theta_{ij} \ge 0$ the integer time offset (TO)
characterizing the $i \to j$ link (encompassing both 
the propagation delay of the wireless link and the processing time at node $i$),  
which models the fact that
the receiver $j$ does not know where the 
multicarrier blocks transmitted by node $i$ start.\sfootnote{The fractional  
TO is incorporated as part of $\{\widetilde{h}_{ij}(\ell)\}_{\ell=0}^{L_{ij}}$.}
We will assume that, for each $i \to j$ link,  
the sum of the channel order and the TO turns 
out to be within one PU symbol, 
i.e., $L_{ij}+\theta_{ij} \le P-1$,  such that 
the desired block received by node $j$ is impaired only 
by the interblock interference (IBI) of the previous block. 
We also assume that  each node is
able to align its local oscillator to the
carrier frequency of the received signal
with negligible error. 
Hereinafter, with reference to a single PU symbol period,
the frequency-domain channel matrix\sfootnote{An $n \times m$ matrix and a column vector 
over the field
$\setF$ are denoted
as $\A \in \setF^{n \times m}$ and $\a \in \setF^{n}$,
respectively; common fields are those of complex, real,
and integer numbers, denoted
with $\setC$, $\setR$, and $\setZ$,
respectively;
$\A^\trasp$, $\A^\herm$, $\A^{-1}$, $\A^{\dag}$, and $\A^{-}$ denote
the transpose, the conjugate transpose,
the inverse, the Moore-Penrose generalized
inverse \cite{Penrose}, and the generalized
(1)-inverse \cite{Penrose} of $\A$, respectively;
$\ob_{m} \in \setR^{m}$,
$\Ob_{m \times n} \in \setR^{m \times n}$,
and $\I_m \in \setR^{m \times m}$ denote
the zero vector, the zero
matrix, and the identity matrix, respectively;
for $\a \in \setC^m$,
$\A = \diag(\a) \in \setC^{m \times m}$ denotes the diagonal matrix
whose diagonal elements are the entries of $\a$;
$\Sf \in \setR^{n \times n}$ and $\Sb \in \setR^{n \times n}$
denote the Toeplitz ``forward shift" and ``backward shift"
matrices \cite{Horn}, respectively, 
where the first column of $\Sf$ and the first row 
of $\Sb$ are given by $[0, 1, 0,  \ldots, 0]^T$
and $[0, 1, 0, \ldots, 0]$, respectively;
$\Vert \a \Vert$ is the Euclidean norm of $\a \in \setC^m$;
$\rank(\B)$ is the rank of $\B \in \setC^{m \times n}$;
$\det(\B)$ denotes the determinant
of $\B\in \setC^{m \times m}$;
$\Prob(A)$ denotes the probability of the event $A$;
the operator $\expect[\cdot]$ denotes ensemble mean
and $\expect[\cdot \,|\, A]$ is the conditional mean given
the event $A$;
finally,  $x^{+} \eqdef \max(x,0)$.
}
\be
\Hcal_{ij} \eqdef
\diag[H_{ij}(0), H_{ij}(1), \ldots, H_{ij}(M-1)]
\in \setC^{M \times M}
\label{eq:H-def}
\ee
with 
\be
H_{ij}(m) \eqdef e^{-j \frac{2 \pi}{M} \theta_{ik} m}
\sum_{\ell=0}^{L_{ik}} \widetilde{h}_{ij}(\ell) \, e^{-j \frac{2 \pi}{M} \ell m} \:,
\quad
\text{for $m \in \mathcal{M} \eqdef \{0,1,\ldots, M-1\}$}
\label{eq:single-DFT}
\ee
collects the $M$-point discrete Fourier transform (DFT) of the {\em extended} channel
impulse response $\widetilde{h}_{ij}(\ell-\theta_{ij})$ corresponding
to the $i \to j$ link, with
$\Hcal_{i_1j_1}$ statistically independent
of $\Hcal_{i_2j_2}$ for $i_1 \neq i_2$ and 
$j_1 \neq j_2$; moreover,
the diagonal entries of
${\Hcal}_{ij}$ are independent and identically  distributed (i.i.d.)
zero-mean circularly symmetric complex
Gaussian (ZMCSCG) random variables (RVs) having variance
$\sigma_{ij}^2$, which
depends on the average
path loss associated to the underlying link.
In the PU symbol period $[n T_{\PU}, (n+1) T_{\PU})$, 
with $n \in \setZ$, the vector
$\widetilde{\v}_{j}(n) \in \setC^{P}$
models the thermal noise at the $j$th receiver, 
with $j \in \{2,3,4\}$.
We assume that $\widetilde{\v}_{j}(n)$ is a ZMCSCG random vector, with
correlation matrix $\expect[\widetilde{\v}_{j}(n) \,  \widetilde{\v}_{j}^\herm(n)] = \sigma^2_{v_j} \, \I_{P}$ and $\widetilde{\v}_{j_1}(n_1)$ statistically independent
of $\widetilde{\v}_{j_2}(n_2)$ for $j_1 \neq j_2$ and $n_1 \neq n_2$.
Finally, channel matrices,
data transmitted by PU and SU,
and noise vectors are statistically independent random objects.

\subsection{Signal transmitted by the PTx}

During the PU symbol period $[n T_{\PU}, (n+1) T_{\PU})$,
the PTx transmits a  
frequency-domain symbol block 
$\x_{\PU}(n) \eqdef [x_{\PU}^{(0)}(n), x_{\PU}^{(1)}(n), \ldots, 
x_{\PU}^{(Q-1)}(n)]^\trasp \in \setC^{Q}$ of $Q$ symbols, modeled as
i.i.d. zero-mean circularly symmetric complex RVs with variance
$\PPU$,
where $\PPU>0$ is the PU power budget.
We assume that CSI 
is not available at the PTx and, hence,
$\PPU$ is
uniformly allocated across all data subcarriers.
Vector $\x_{\PU}(n)$ is augmented by
VCs insertion in arbitrary
positions $\IPUvc \eqdef \{q_0, q_1, \ldots, q_{\Mvc-1}\}$,
thus obtaining the block 
$\bm{\Theta} \, \x_{\PU}(n)$, with
$\bm{\Theta} \in \setR^{M \times Q}$ modeling
VCs insertion.
Matrix $\bm{\Theta}$ or, equivalently, the set $\IPUvc$ 
can be statically specified 
by the standard, or it can be dynamically adjusted 
to select the best $Q$ available subcarriers, i.e., those  with  the  highest
signal-to-noise ratios (SNRs).
Then, the block $\bm{\Theta} \, \x_{\PU}(n)$ is subject to conventional
OFDM processing, encompassing
$M$-point inverse discrete Fourier
transform (IDFT), followed by CP insertion, thus obtaining
(see, e.g., \cite{Wang})
$\widetilde{\u}_{\PU}(n) = \Tcp \, \Widft \, \bm{\Theta} \, \x_{\PU}(n)$,
where 
$\Tcp \eqdef [\I_{\text{cp}}^\trasp, \I_M]^\trasp \in \setR^{P \times M}$,
with $\I_{\text{cp}} \in \setR^{\Lcppu \times M}$
obtained from $\I_M$ by picking its last $\Lcppu$ rows, 
and $\Widft \in \setC^{M \times M}$ is
the unitary symmetric IDFT matrix \cite{Wang}.
The entries of $\widetilde{\u}_{\PU}(n)$ are subject to 
digital-to-analog (D/A) plus radio-frequency (RF) conversion for transmission over the wireless channel.

\subsection{Signal transmitted by the STx}
\label{sec:STx-trans}

Let $\widetilde{y}_{2}^{(p)}(n)$ denote 
the baseband-equivalent $p$th sample received 
by the STx within the $n$th PU symbol period,
for $p \in \mathcal{P} \eqdef \{0,1,\ldots, P-1\}$.
By gathering such samples in the vector $\widetilde{\y}_2(n) \eqdef [\widetilde{y}_2^{(0)}(n), \widetilde{y}_2^{(1)}(n), 
\ldots, \widetilde{y}_2^{(P-1)}(n)]^\trasp \in \setC^{P}$, the received signal can be 
expressed as
\be
\widetilde{\y}_2(n) = \widetilde{\H}_{12}^{(0)} \, \widetilde{\u}_{\PU}(n)
+ \widetilde{\H}_{12}^{(1)} \, \widetilde{\u}_{\PU}(n-1) + \widetilde{\v}_{2}(n)
\label{eq:vetr2}
\ee
where we remember that $\widetilde{\v}_{2}(n)$ is the noise vector, whereas  
\barr
\widetilde{\H}_{12}^{(0)} & \eqdef \sum_{\ell=0}^{L_{12}} \widetilde{h}_{12}(\ell) \, \Sf^{\ell+ \theta_{12}} \in \setC^{P \times P}
\label{eq:C120}
\\
\widetilde{\H}_{12}^{(1)} & \eqdef \sum_{\ell=0}^{L_{12}} \widetilde{h}_{12}(\ell) \, \Sb^{P-\ell- \theta_{12}} \in \setC^{P \times P}
\label{eq:C121}
\earr
are Toeplitz lower- and upper-triangular matrices, respectively.

In the proposed spectrum sharing scheme, the STx exploits the $n$th PU transmission to deliver to the SRx
a frequency-domain block $\x_{\SU}(n) \eqdef [x_{\SU}^{(0)}(n), 
x_{\SU}^{(1)}(n), \ldots, x_{\SU}^{(N+\Mvc-1)}(n)]^\trasp \in \setC^{N+\Mvc}$, which is composed of $N+\Mvc$ symbols, modeled as 
i.i.d. zero-mean
unit-variance circularly symmetric complex RVs.
It is assumed that $N +\Mvc \le M$ and, thus, the {\em rate} of the
SU is $N +\Mvc$ symbols per OFDM block of the PU.
Specifically, the block 
$\widetilde{\z}_2(n) \in \setC^{P}$, 
transmitted by the STx during the 
$n$th PU symbol period, is composed of two summands
$\widetilde{\z}_2(n)=\widetilde{\z}_{2,\text{I}}(n)
+\widetilde{\z}_{2,\text{II}}(n)$:
the former one $\widetilde{\z}_{2,\text{I}}(n)$ conveys the symbols
$\xsuuc(n)\eqdef [x_{\SU}^{(0)}(n), 
x_{\SU}^{(1)}(n), \ldots, x_{\SU}^{(N-1)}(n)]^\trasp \in \setC^N$
to be transmitted over the $Q$ used subcarriers of  the PU,
whereas the latter one
$\widetilde{\z}_{2,\text{II}}(n)$ is a linear transformation of the symbols
$\xsuvc(n)\eqdef [x_{\SU}^{(N)}(n), 
x_{\SU}^{(N+1)}(n), \ldots, x_{\SU}^{(N+\Mvc-1)}(n)]^\trasp \in \setC^\Mvc$ to be sent over the $\Mvc$ VCs of the PU. 
We note that $\x_{\SU}(n) =[\xsuuc^\trasp(n), \xsuvc^\trasp(n)]^\trasp$.

The first summand 
$\widetilde{\z}_{2,\text{I}}(n)
\eqdef [\widetilde{z}_{2,\text{I}}^{(0)}(n), \widetilde{z}_{2,\text{I}}^{(1)}(n), 
\ldots, \widetilde{z}_{2,\text{I}}^{(P-1)}(n)]^\trasp \in \setC^P$ is obtained by 
performing a linear
transformation of the received block $\widetilde{\y}_2(n)$ through the 
Toeplitz lower-triangular matrix 
\be
\widetilde{\F}(n)  \eqdef \sum_{p=0}^{\Lpre} \widetilde{f}^{(p)}(n) \, \Sf^{p} 
\in \setC^{P \times P} \:,
\quad
\text{with $\Lpre < M$}
\label{eq:Ftilde}
\ee
that is, $\widetilde{\z}_{2,\text{I}}(n) = \widetilde{\F}(n) \, \widetilde{\y}_2(n)$,
where $\{\widetilde{f}^{(p)}(n)\}_{p=0}^{\Lpre}$ piggybacks the 
symbols in $\xsuuc(n)$. 
We note that $\widetilde{\z}_{2,\text{I}}(n)$ depends on the received signal $\widetilde{\y}_2(n)$ and, thus, it must be computed in {\em real-time}.
In this regard, it is noteworthy that, for each $n \in \setZ$, the block $\widetilde{\z}_{2,\text{I}}(n)$ 
can be interpreted as the output of a discrete-time causal LTI filter having
$\widetilde{f}^{(p)}(n)$ and $\widetilde{y}_{2}^{(p)}(n)$ as impulse response and 
input signal, respectively, i.e., 
\be
\widetilde{z}_{2,\text{I}}^{(p)}(n) = \sum_{\ell=0}^p \widetilde{f}^{(p-\ell)}(n) \, \widetilde{y}_{2}^{(\ell)}(n) \: ,
\quad \text{for $p \in \mathcal{P}$} 
\label{eq:direct-conv}
\ee
supposing that $\widetilde{f}^{(\ell)}(n)=0$ for $\ell <0$.
The functional dependence of the matrix $\widetilde{\F}(n)$ on 
the symbol subvector $\xsuuc(n)$
is much easier to explain in the frequency-domain. For each $n \in \setZ$, let 
$\f(n) \eqdef [F^{(n)}(0), F^{(n)}(1), \ldots, F^{(n)}(M-1)]^\trasp \in \setC^M$,
with 
\be
F^{(n)}(m) \eqdef 
\sum_{\ell=0}^{\Lpre} \widetilde{f}^{(\ell)}(n) 
\, e^{-j \frac{2 \pi}{M} \ell m} \: ,
\label{eq:single-DFT-pre}
\quad
\text{for $m \in \mathcal{M}$}
\ee
being the $M$-point DFT of 
$\J \, \widetilde{\f}(n)$,
where $\J \eqdef [\I_{\Lpre+1}, \mathbf{O}_{(M-\Lpre-1) \times (\Lpre+1)}^\trasp]^\trasp
\in \setR^{M \times (\Lpre+1)}$ is a zero-padding matrix and 
$\widetilde{\f}(n) \eqdef [\widetilde{f}^{(0)}(n), \widetilde{f}^{(1)}(n), \ldots, \widetilde{f}^{(\Lpre)}(n)]^\trasp \in \setC^{\Lpre+1}$ completely describes the matrix $\widetilde{\F}(n)$  given by \eqref{eq:Ftilde}.
In our scheme, we impose that $F^{(m)}(n) = \bm{\alpha}^\herm_m 
\, \xsuuc(n)$ is a linear combination of the SU symbols, with
$\bm{\alpha}_m \eqdef [\alpha_{m,0}, \alpha_{m,1}, \ldots, \alpha_{m,N-1}]^\trasp \in \setC^N$.
In this case, it results that
$\f(n) = \Acal \, \xsuuc(n)$, 
where 
$\Acal \eqdef [\bm{\alpha}_0, \bm{\alpha}_1, \ldots, \bm{\alpha}_{M-1}]^\herm \in \setC^{M \times N}$ is a {\em frequency-domain precoding matrix} of the SU symbols. In order to ensure that $\f(n)=\mathbf{0}_M$ iff $\xsuuc(n)=\mathbf{0}_N$, the matrix $\Acal$ must be full-column rank, 
i.e., $\rank(\Acal)=N$.

At this point, let us focus on the second summand 
$\widetilde{\z}_{2,\text{II}}(n) \in \setC^P$. Such a vector 
does not depend on the received data $\widetilde{\y}_2(n)$
and, hence,  it is already available at the beginning of 
the time interval 
$[n T_{\PU}, (n+1) T_{\PU})$. 
In our scheme,  it is generated as 
$\widetilde{\z}_{2,\text{II}}(n) \equiv \widetilde{\u}_{\SU}(n) 
\eqdef \Tcp \, \Widft \, \Gcal \, \xsuvc(n)$, 
where $\Gcal \eqdef [\bm{\gamma}_0, \bm{\gamma}_1, \ldots, \bm{\gamma}_{M-1}]^\herm \in \setC^{M \times \Mvc}$ is
{\em another} frequency-domain precoding matrix of the SU symbols, with $\bm{\gamma}_m \in \setC^\Mvc$, whose choice will be clear in 
Subsection~\ref{sec:SRx} and 
Section~\ref{sec:analysis-su}.\sfootnote{\label{foot:not}For each $m \in \mathcal{M}$, the 
weight vectors 
$\bm{\alpha}_m$ and $\bm{\gamma}_m$ might change
from one symbol to another. For the sake of simplicity, 
we will not explicitly indicate the dependence of $\Acal$ 
and $\Gcal$
on the PU symbol period.}
Therefore, the overall time-domain data block transmitted by the STx is given by
\be
\widetilde{\z}_{2}(n) = \widetilde{\F}(n) \, \widetilde{\y}_2(n) + \widetilde{\u}_{\SU}(n) 
\label{eq:pre-total}
\ee
whose entries are subject to D/A plus RF conversion for transmission over the wireless channel.
The main signal processing operations carried out by the STx
are depicted in Fig.~\ref{fig:figure_2}. 
Strictly speaking, the STx acts as a full-duplex AF relay, 
which linearly processes the received data $\widetilde{\y}_2(n)$ through
an own information-bearing matrix $\widetilde{\F}(n)$, 
by also adding the term $\widetilde{\u}_{\SU}(n)$. 

\begin{figure}[tbp]
\centering
\psfrag{inp\r}{}
\psfrag{y2tilde\r}{$\widetilde{\y}_{2}(n)$}
\psfrag{z2tilde\r}{$\widetilde{\z}_{2}(n)$}
\psfrag{z2tildeI\r}{$\widetilde{\z}_{2,\text{I}}(n)$}
\psfrag{z2tildeII\r}{$\widetilde{\u}_{\SU}(n) $}
\psfrag{Fmat\r}{\!$\widetilde{\F}(n)$}
\psfrag{Gmat\r}{$\,\,\,\Gcal$}
\psfrag{xsuI\r}{$\xsuuc(n)$}
\psfrag{xsuII\r}{$\xsuvc(n)$}
\includegraphics[width=0.80\linewidth, trim = 0 0 0 0]{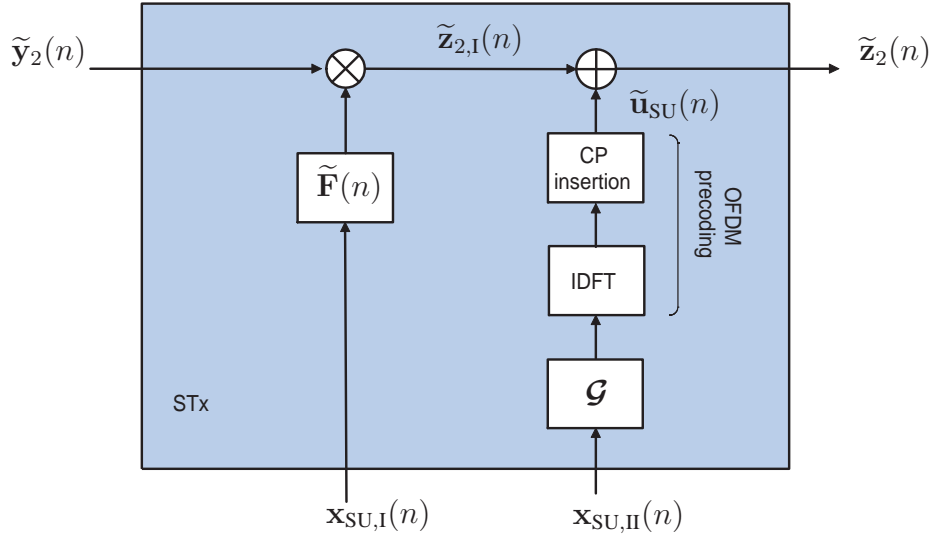}
\caption{Baseband processing carried out at the secondary user transceiver.}
\label{fig:figure_2}
\end{figure}

Implementation of the convolution formula \eqref{eq:direct-conv}
requires the synthesis of the time-domain vector $\widetilde{\f}(n)$.
The relationship between $\widetilde{\f}(n)$ and its frequency-domain counterpart $\f(n)$ is given by 
\be
\sqrt{M} \, \Wdft \, \J \, \widetilde{\f}(n) =\f(n)
\label{eq:syst}
\ee
with 
$\Wdft \eqdef \Widft^{-1}=\Widft^\herm$
defining the unitary symmetric DFT matrix \cite{Wang}.
\begin{lemma}
\label{lem:consist}
The system of linear equations \eqref{eq:syst} is consistent 
(i.e., it admits at least one solution) iff
$\f(n)= \bm{\Pi}_{\text{IDFT}} \, \r(n)$, where 
the columns of $\bm{\Pi}_{\text{IDFT}} \in \setC^{M \times (\Lpre+1)}$ 
form a basis for the null space of $\Widftdown$,
i.e., $\Widftdown \, \bm{\Pi}_{\text{IDFT}}=\mathbf{O}_{(M-\Lpre-1) \times (\Lpre+1)}$,\sfootnote{It can be assumed, without loss of generality, that $\bm{\Pi}_{\text{IDFT}}$ 
is semi-unitary i.e., $\bm{\Pi}_{\text{IDFT}}^\herm \, \bm{\Pi}_{\text{IDFT}}=\I_{\Lpre+1}$.}
with $\Widftdown \in \setC^{(M-\Lpre-1) \times M}$
obtained from $\Widft$ by picking its last $M-\Lpre-1$ rows, 
and $\r(n) \in \setC^{\Lpre+1}$ is an arbitrary vector.
\end{lemma}
\begin{proof}
System \eqref{eq:syst} is consistent \cite{Penrose}
iff
$(\sqrt{M} \, \Wdft \, \J) \, (\sqrt{M} \, \Wdft \, \J)^{-} \, \f(n) = \f(n)$. Since $(\sqrt{M} \, \Wdft \, \J)^{-}= \J^{-} \, \Widft/\sqrt{M}$,
with $\J^{-}=\J^\trasp$, the previous equation can be equivalently written 
after straightforward algebraic manipulations 
as $(\I_M-\J \, \J^\trasp) \, \Widft \, \f(n)=\mathbf{0}_M$ which, accounting for the structure of $\J$ and partitioning $\Widft$ accordingly, leads to the homogeneous system of linear equations
$\Widftdown \, \f(n)=\mathbf{0}_{M-\Lpre-1}$. Hence, system 
\eqref{eq:syst} is consistent iff $\f(n)$ belongs to the null space of
$\Widftdown$.  The proof ends by observing that $\Widftdown$
is full-row rank, i.e., $\rank(\Widftdown)=M-\Lpre-1$, and, thus, 
the dimension of its null space (i.e., its nullity) is equal to $M-\rank(\Widftdown)=\Lpre+1$.
\end{proof}

As an immediate consequence of Lemma~\ref{lem:consist}, one has that the
precoding matrix $\Acal$ cannot be completely arbitrary but, instead, it must obey 
$\Acal=\bm{\Pi}_{\text{IDFT}} \, \Acalover$, with $\Acalover \in \setC^{(\Lpre+1) \times N}$.
It is important to note that, in this case, the rank condition $\rank(\Acal)=N$
mandates $\rank(\bm{\Pi}_{\text{IDFT}} \, \Acalover)=N$, which happens iff 
$\rank(\Acalover)=N \le \Lpre+1$,\sfootnote{It results \cite{Horn} that $\rank(\Acalover) \le \rank(\bm{\Pi}_{\text{IDFT}} \, \Acalover) \le \min[\Lpre+1, \rank(\Acalover)]$.} i.e., the length of the vector $\xsuuc(n)$
cannot be greater than the filter length $\Lpre+1$.
In this case, the {\em minimal-norm} solution of 
\eqref{eq:syst} reads as
\be
\widetilde{\f}(n) = (\sqrt{M} \, \Wdft \, \J )^\dag \, \f(n) =
\frac{1}{\sqrt{M}} \, \J^\trasp \, \Widft \, \bm{\Pi}_{\text{IDFT}} \, \Acalover \, \xsuuc(n) \: .
\label{eq:minnormsol}
\ee
The choice of $\Acalover$  may depend on
the available CSI at the STx
and will be discussed 
in Subsection~\ref{sec:SRx} and Section~\ref{sec:analysis-su}.

As a final remark, the (linear) convolution \eqref{eq:direct-conv} can be directly calculated in real-time without inherent 
latency: indeed, if computations were instantaneous, each sample 
of the received signal would yield a corresponding output
to be transmitted by the STx. The actual latency of the direct convolution 
\eqref{eq:direct-conv} results from the time necessary
to compute each output sample. We will see in Subsections~\ref{sec:PRx}
and \ref{sec:SRx} that, due to other constraints, 
the filter order $\Lpre$ has to be smaller than $\Lcppu$. The computation time is shorter than the sampling period $T_c$ and, thus, 
the latency can be assumed to be equal to one sample only.
Even though the computational cost of \eqref{eq:direct-conv} increases 
linearly with the filter order,  
with respect to frequency-domain block convolution
techniques based on the fast Fourier transform (FFT),\sfootnote{Computationally efficient FFT-based methods cannot be used 
for evaluating the linear convolution \eqref{eq:direct-conv}
since they have an inherent input-to-output
latency equal to the length $P$ of the block, i.e., one symbol
period $T_\PU$: indeed, the input block $\widetilde{\y}_2(n)$ 
must be fully available in order to start computing the 
the samples of the output block $\widetilde{\z}_{2,\text{I}}(n) $.} 
the price to pay in terms of computational complexity is negligible
when $\Lpre < \Lcppu \ll M$.

\subsection{Signal received and processed by the PRx}
\label{sec:PRx}

Let $\widetilde{\y}_3(n) \in \setC^{P}$ gather 
the baseband-equivalent samples received 
by the PRx within the $n$th PU symbol period.
Accounting for \eqref{eq:vetr2} and 
\eqref{eq:pre-total},  one gets
\begin{multline}
\widetilde{\y}_3(n) = \widetilde{\H}_{13}^{(0)} \, \widetilde{\u}_{\PU}(n)
+ \widetilde{\H}_{13}^{(1)} \, \widetilde{\u}_{\PU}(n-1) + 
\widetilde{\H}_{23}^{(0)} \, \widetilde{\z}_2(n)
+ \widetilde{\H}_{23}^{(1)} \, \widetilde{\z}_2(n-1) + \widetilde{\v}_{3}(n)
\\ =  \left[\widetilde{\H}_{13}^{(0)} + \widetilde{\H}_{23}^{(0)} \, \widetilde{\F}(n)  \, 
\widetilde{\H}_{12}^{(0)} \right] \widetilde{\u}_{\PU}(n)
+ \left[ \widetilde{\H}_{13}^{(1)} + 
\widetilde{\H}_{23}^{(0)} \, \widetilde{\F}(n) \, \widetilde{\H}_{12}^{(1)} +
\widetilde{\H}_{23}^{(1)} \, \widetilde{\F}(n-1) \, \widetilde{\H}_{12}^{(0)}
\right] \widetilde{\u}_{\PU}(n-1) 
\\ + \widetilde{\H}_{23}^{(1)} \, \widetilde{\F}(n-1) \, \widetilde{\H}_{12}^{(1)}
\, \widetilde{\u}_{\PU}(n-2) +
\widetilde{\H}_{23}^{(0)} \, \widetilde{\u}_{\SU}(n)
+ \widetilde{\H}_{23}^{(1)} \, \widetilde{\u}_{\SU}(n-1) 
\\ +
\widetilde{\H}_{23}^{(0)} \, \widetilde{\F}(n) \, \widetilde{\v}_{2}(n) +
\widetilde{\H}_{23}^{(1)} \, \widetilde{\F}(n-1) \, \widetilde{\v}_{2}(n-1) 
+ \widetilde{\v}_{3}(n) 
\label{eq:vetr3}
\end{multline}
where $\{\widetilde{\H}_{13}^{(0)}, 
\widetilde{\H}_{13}^{(1)} \}$ and  
$\{\widetilde{\H}_{23}^{(0)}, 
\widetilde{\H}_{23}^{(1)}\}$ can be obtained from
\eqref{eq:C120} and \eqref{eq:C121} 
by replacing $\{L_{12}, \widetilde{h}_{12}(\ell), \theta_{12}\}$
with $\{L_{13}, \widetilde{h}_{13}(\ell), \theta_{13}\}$ 
and $\{L_{23}, \widetilde{h}_{23}(\ell), 
\theta_{23}\}$, respectively, and
we remember that $\widetilde{\v}_{3}(n)$ accounts for noise.

The product of any lower (upper) triangular Toeplitz matrices is a lower (upper) triangular Toeplitz matrix, too \cite{Horn}. 
Indeed, it is directly verified that, if the following inequality
\be
L_{12}+\Lpre+L_{23}+ \theta_{12}+\theta_{23} \le P-1
\label{eq:dis}
\ee
holds,  
the product $\widetilde{\H}_{23}^{(0)} \, \widetilde{\F}(n)  \, 
\widetilde{\H}_{12}^{(0)}$ is a  
lower-triangular Toeplitz matrix 
having as first column
$[\bm{0}_{\theta_{12}+\theta_{23}}^\trasp, \widetilde{\h}_{123}^\trasp(n), 
\bm{0}^\trasp_{P-L_{12}-\Lpre-L_{23}-\theta_{12}-\theta_{23}-1}]^\trasp$,
where the vector $\widetilde{\h}_{123}\in \setC^{L_{12}+\Lpre+L_{23}+1}$ collects the samples 
of the (linear) convolution
among $\{\widetilde{h}_{12}(\ell)\}_{\ell=0}^{L_{12}}$, 
$\{\widetilde{f}^{(\ell)}(n)\}_{\ell=0}^{\Lpre}$, 
and $\{\widetilde{h}_{23}(\ell)\}_{\ell=0}^{L_{23}}$.
Moreover, one has
$\widetilde{\H}_{23}^{(1)} \, \widetilde{\F}(n-1) \, \widetilde{\H}_{12}^{(1)} =\mathbf{O}_{P \times P}$, provided that 
\eqref{eq:dis} is fulfilled.
On the other hand, it is verified by direct inspection that: (i) only
the first $L_{12}+\Lpre +L_{23}+ \theta_{12}+\theta_{23}$ rows
of $\widetilde{\H}_{23}^{(0)} \, \widetilde{\F}(n) \, \widetilde{\H}_{12}^{(1)}$ might
not be zero; (ii) the last $P-L_{i3}-\theta_{i3}$ rows of the matrix
$\widetilde{\H}_{i3}^{(1)}$ are identically zero, for $i \in \{1,2\}$;
(iii) the nonzero entries of  $\widetilde{\H}_{23}^{(1)} \, \widetilde{\F}(n-1) \, \widetilde{\H}_{12}^{(0)}$ and 
$\widetilde{\H}_{23}^{(1)} \, \widetilde{\F}(n-1) \, \widetilde{\v}_{2}(n-1)$
are located within their first $L_{23}+\theta_{23}$ rows.
Therefore, if the CP is designed such that 
\be
\Lcppu \ge \max\left(L_{12}+\Lpre+L_{23} + \theta_{12}+\theta_{23}, 
L_{13}+\theta_{13}\right)
\label{eq:CP}
\ee
the IBI contribution in \eqref{eq:vetr3} can be completely discarded by dropping the first $\Lcppu$ components of $\widetilde{\y}_3(n)$. 
In other words, the convolutive process carried out by the STx
may increase the frequency selectivity of the end-to-end PU channel. 
This drawback can be overcome by increasing the CP length as in \eqref{eq:CP}, 
which leads to an  inherent reduction of the transmission data rate 
of the PU system when 
$L_{12}+\Lpre+L_{23} + \theta_{12}+\theta_{23} > 
L_{13}+\theta_{13}$. However, such a (possible) loss 
turns out to be negligible if the number $M$
of subcarriers is significantly greater than $\Lcppu$.
Most important, we show in Section~\ref{sec:analysis-pu}
that,  if the legacy system is designed to
fulfil \eqref{eq:CP}, it might even achieve 
a significant performance gain.
Moreover, assumption
\eqref{eq:CP} requires only upper bounds (rather than the
exact knowledge) on the channel orders and TOs. In general, depending on the transmitted
signal parameters (carrier frequency and bandwidth) and environment 
(indoor or outdoor), the maximum channel multipath
spread is known and the TOs  are confined to a small uncertainty interval,
whose support can be typically predicted. 

CP removal is accomplished
by defining the matrix $\R_{\text{cp}}\eqdef [\mathbf{O}_{M \times \Lcppu}, \I_{M}] \in \setR^{M \times P}$ and
forming the product $\R_{\text{cp}} \, \widetilde{\y}_3(n)$.
If \eqref{eq:minnormsol} and \eqref{eq:CP} hold,  after discarding the CP and performing $M$-point DFT, the received signal over all the subcarriers
at the PRx  can be expressed as follows
\be
\y_{\PU}(n) \equiv \y_3(n) \eqdef \Wdft \, \R_{\text{cp}} \, \widetilde{\y}_3(n)
= \Hcal_{\PU}(n) \, \bm{\Theta} \, \x_{\PU}(n) + \v_{\PU}(n)
\label{y_PU}
\ee
where $\Hcal_{\PU}(n) \eqdef \Hcal_{13} + \Hcal_{23} \, \Fxsu(n) \, \Hcal_{12} \in \setC^{M \times M}$ is a diagonal matrix whose $m$th diagonal entry is given by
\be
H_{\PU}(m) \eqdef H_{13}(m)+H_{12}(m) \, H_{23}(m) \, F^{(n)}(m) \: ,
\quad \text{with $m \in \mathcal{M}$}
\label{eq:HPUentry}
\ee
with $F^{(n)}(m)=\{\f(n)\}_m=\{\bm{\Pi}_{\text{IDFT}} \, \Acalover \, \xsuuc(n)\}_m$,
and
\be
\v_{\PU}(n) \eqdef 
\Wdft \, \R_{\text{cp}} \, \widetilde{\H}_{23}^{(0)} \, \widetilde{\F}(n) \, 
\widetilde{\v}_{2}(n) +\Hcal_{23} \, \Gcal \, \xsuvc(n) + \v_{3}(n) \in \setC^M
\label{eq:eq-noise-3}
\ee
represents the {\em equivalent noise} vector at the PRx, the matrices $\Hcal_{13}$ and $\Hcal_{23}$ have been defined in
\eqref{eq:H-def}, whereas
$\Fxsu(n) \eqdef \diag[\f(n)] =\diag[\bm{\Pi}_{\text{IDFT}} \, \Acalover \, 
\xsuuc(n)]
\in \setC^{M \times M}$
collects the $M$-point DFT samples 
\eqref{eq:single-DFT-pre}, 
and $\v_j(n) \eqdef [{v}_j^{(0)}(n), {v}_j^{(1)}(n), \ldots, {v}_j^{(M-1)}(n)]^\trasp = \Wdft \, \R_{\text{cp}} \, \widetilde{\v}_j(n)
\in \setC^{M}$, for $j \in \{2,3\}$.
It is apparent
that the overall PU
channel matrix $\Hcal_{\PU}(n)$
incorporates the contribution of the SU 
symbol block $\x_{\SU}(n)$.
The entries of
$\Hcal_{\PU}(n)$ can be
estimated at the PRx using training symbols
transmitted by the PTx and,
thus, knowledge of $\Fxsu(n)$ is not required at
the PRx.\sfootnote{\label{note-channelestimate}The
channel estimation error
can be made negligible, for intermediate-to-high SNRs,
by using a number of training symbols that is not smaller
than $\Lcppu$ and by
carefully designing the PU training sequence \cite{Mengali}.}

\subsection{Signal received and processed by the SRx}
\label{sec:SRx}

Similarly to \eqref{eq:vetr3},
the baseband-equivalent received data vector 
by the SRx within the $n$th PU symbol period
can be expressed as
\be
\widetilde{\y}_4(n) = \widetilde{\H}_{14}^{(0)} \, \widetilde{\u}_{\PU}(n)
+ \widetilde{\H}_{14}^{(1)} \, \widetilde{\u}_{\PU}(n-1) + 
\widetilde{\H}_{24}^{(0)} \, \widetilde{\z}_2(n)
+ \widetilde{\H}_{24}^{(1)} \, \widetilde{\z}_2(n-1) + \widetilde{\v}_{4}(n)
\label{eq:vetr4}
\ee
where $\{\widetilde{\H}_{14}^{(0)}, 
\widetilde{\H}_{14}^{(1)} \}$ and  
$\{\widetilde{\H}_{24}^{(0)}, 
\widetilde{\H}_{24}^{(1)}\}$ can be obtained from
\eqref{eq:C120} and \eqref{eq:C121} 
by replacing $\{L_{12}, \widetilde{h}_{12}(\ell), \theta_{12}\}$
with $\{L_{14}, \widetilde{h}_{14}(\ell), \theta_{14}\}$ 
and $\{L_{24}, \widetilde{h}_{24}(\ell), 
\theta_{24}\}$, respectively, and 
we remember that 
$\widetilde{\v}_{4}(n)$ is the noise vector.
Paralleling the same arguments of Subsection~\ref{sec:PRx}, it can 
be shown that, if $L_{12}+\Lpre+L_{24}+ \theta_{12}+\theta_{24} \le P-1$
and\sfootnote{To suppress only its own IBI represented by $\widetilde{\u}_{\SU}(n-1)$,  the SRx could even discard a portion of the received data 
smaller than $\Lcppu$, thus accepting the IBI of the PU transmission
due to $\widetilde{\u}_{\PU}(n-1)$. However, in this case, more complex receiving structures would be required to reliably estimate the desired symbol block $\x_{\SU}(n)$ \cite{Darsena2007a,Darsena2015}.
} 
\be
\Lcppu \ge \max\left(L_{12}+\Lpre+L_{24} + \theta_{12}+\theta_{24}, 
L_{14}+\theta_{14}\right)
\label{eq:CP-2}
\ee
after discarding the CP and performing $M$-point DFT,
the frequency-domain signal received at the SRx can be written as
\be
\y_{\SU}(n) \equiv \y_4(n) \eqdef \Wdft \, \R_{\text{cp}} \, \widetilde{\y}_4(n)=
\Hcal_{\SU}(n) \, \bm{\Delta} \, \x_{\SU}(n) + \v_{\SU}(n)
\label{y_SU}
\ee
where
${\Hcal}_{\SU}(n)  \eqdef [ \H_{\SU}(n), \Hcal_{24}] \in \setC^{M \times 2 M}$, 
$\H_{\SU}(n)  \eqdef
\Hcal_{24} \, [\Hcal_{12} \, \X_{\PU}(n) +\V_2(n)] \in \setC^{M \times M}$,
$\bm{\Delta}  \eqdef \diag(\bm{\Pi}_{\text{IDFT}} \, \Acalover, \Gcal) \in \setC^{2 M \times (N+\Mvc)}$ represents  the {\em overall} 
frequency-domain precoding matrix of the SU,
and $\v_{\SU}(n)  \eqdef \Hcal_{14} \, \bm{\Theta} \, \x_{\PU}(n)+ \v_{4}(n) \in \setC^M$ denotes the  {\em equivalent noise} term at the SRx.
Additionally,  $\X_{\PU}(n) \eqdef \diag[\bm{\Theta} \, \x_{\PU}(n)] \in \setC^{M \times M}$,
$\V_{2}(n) \eqdef \diag[\v_2(n)] \in \setC^{M \times M}$, 
$\v_4(n) \eqdef \Wdft \, \R_{\text{cp}} \, \widetilde{\v}_4(n) \in \setC^{M}$,
and the diagonal channel matrices $\Hcal_{14}$ and $\Hcal_{24}$ have been defined in \eqref{eq:H-def}.
In writing \eqref{y_SU}, we have replaced the noise vector 
$\widetilde{\v}_{2}(n)$ with $\hat{\v}_{2}(n) \eqdef \Tcp \, \v_2(n)$: they are 
both ZMCSCG random vectors with correlation matrix 
$\expect[\widetilde{\v}_{2}(n) \,  \widetilde{\v}_{2}^\herm(n)] = \sigma^2_{v_2} \, \I_{P}$ and $\expect[\hat{\v}_{2}(n) \,  \hat{\v}_{2}(n)^\herm] = \sigma^2_{v_2} \, \Tcp \, \Tcp^\trasp$, respectively. 
For sufficiently large values of $M$, the matrices $\I_{P}$
and $\Tcp \, \Tcp^\trasp$ are asymptotically equivalent in weak norm \cite{Gray_book}. Therefore, in the large $M$ limit, 
the random vectors $\widetilde{\v}_{2}(n)$ and $\hat{\v}_{2}(n)$
have the same distribution.
Moreover, it is noteworthy that
the channel matrix $\H_{\SU}(n)$  is a diagonal matrix, whose $m$th diagonal entry is given by
\be
H_{\SU}^{(n)}(m) = H_{24}(m)  \, [H_{12}(m) \, x_{\PU}^{(m)}(n) \, \beta_m + v_2^{(m)}(n)]
\label{eq:diag-hsu}
\ee
where $\beta_m=0$ if $m \in \IPUvc$, whereas 
$\beta_m=1$ if $m \in \IPUuc \eqdef \{0,1,\ldots,M-1\}-\IPUvc$,
which represents the set of PU used subcarriers. 
The ``composite'' matrix
${\Hcal}_{\SU}(n)$ in \eqref{y_SU} can be reliably
estimated at the SRx using training symbols
transmitted by the STx (footnote~\ref{note-channelestimate}
also applies in this case with obvious modifications).
Signal models \eqref{y_PU} and \eqref{y_SU} hold 
if the CP of the PU system is designed
to satisfy both inequalities \eqref{eq:CP} and \eqref{eq:CP-2}.
Such an assumption
is quite reasonable when the STx is very close to the PTx,
which is the network scenario where our proposed scheme
ensures a significant performance gain for the PU system
(see Section~\ref{sec:analysis-pu}).

Some preliminary comments are now in order regarding the
choice of the precoding matrices $\Acal= \bm{\Pi}_{\text{IDFT}} \, \Acalover$ and $\Gcal$.
Since $H_{\SU}^{(n)}(m) = H_{24}(m)  \, v_2^{(m)}(n)$ for $m \in \IPUvc$,
over the PU VCs, the time-domain convolution \eqref{eq:direct-conv}
leads to a multiplicative superposition of
the SU symbols on the noise samples $\{v_2^{(m)}(n)\}_{m \in \IPUvc}$: 
it is intuitive that, from the SU viewpoint,
such a strategy is detrimental for vanishingly
small noise variances. 
Therefore, any reasonable optimization criterion of the SU precoder 
will impose that
$\bm{\alpha}_{q_0} = \bm{\alpha}_{q_1} = \cdots =
\bm{\alpha}_{q_{\Mvc-1}}=\mathbf{0}_N$.
Since 
the $m$th row of $\Acal$ is given by $\bm{\alpha}_{m}^\herm= \, [\bm{\pi}_{\text{IDFT}}^{(m)}]^\herm \, \Acalover$,  where the conjugate transpose of 
$\bm{\pi}_{\text{IDFT}}^{(m)} \in \setC^{\Lpre+1}$ is the
$m$th row of $\bm{\Pi}_{\text{IDFT}}$,
for $m \in \mathcal{M}$, 
such a condition is tantamount to 
the matrix equation $\Acalover^\herm \, {\bm{\Pi}}_{\text{vc}}=\mathbf{O}_{N \times \Mvc}$,
with ${\bm{\Pi}}_{\text{vc}} \eqdef [\bm{\pi}^{(q_0)}_{\text{IDFT}}, \bm{\pi}^{(q_1)}_{\text{IDFT}}, 
\ldots, \bm{\pi}^{(q_{\Mvc-1})}_{\text{IDFT}}] \in \setC^{(\Lpre+1) \times \Mvc}$, whose general solution \cite{Penrose} can be written as 
\be
\Acalover = \bm{\Upsilon}_{\text{vc}} \, \Ccal
\quad
\Rightarrow
\quad
\Acal = \bm{\Pi}_{\text{IDFT}} \, \bm{\Upsilon}_{\text{vc}} \, \Ccal
\label{eq:C}
\ee
where the columns of $\bm{\Upsilon}_{\text{vc}} \in \setC^{(\Lpre+1) \times (\Lpre-R_{\text{vc}}+1)}$ form a basis 
for the null space of ${\bm{\Pi}}_{\text{vc}}^\herm$,
i.e., ${\bm{\Pi}}_{\text{vc}}^\herm \, \bm{\Upsilon}_{\text{vc}}=\mathbf{O}_{\Mvc \times (\Lpre-R_{\text{vc}}+1)}$,\sfootnote{It can be assumed, without loss of generality, that $\bm{\Upsilon}_{\text{vc}}$ 
is semi-unitary i.e., $\bm{\Upsilon}_{\text{vc}}^\herm \, 
\bm{\Upsilon}_{\text{vc}}=\I_{\Lpre-R_{\text{vc}}+1}$.}
$\Ccal \in \setC^{(\Lpre-R_{\text{vc}}+1) \times N}$ is an arbitrary matrix
to be designed, and $R_{\text{vc}} \eqdef \rank(\bm{\Pi}_{\text{vc}}) = \min(\Lpre+1, \Mvc)$.
Remembering $\rank(\Acalover)=N \le \Lpre+1$, it follows from 
\eqref{eq:C} that $\rank(\bm{\Upsilon}_{\text{vc}} \, \Ccal)=N$, which happens iff $\rank(\Ccal)=N \le \Lpre-R_{\text{vc}}+1$.\sfootnote{It results \cite{Horn} that $\rank(\Ccal) \le \rank(\bm{\Upsilon}_{\text{vc}} \, \Ccal) \le 
\min[\Lpre-R_{\text{vc}}+1, \rank(\Ccal)]$.} Factorization
\eqref{eq:C} further reduces the number $N$ of symbols 
that the SU can transmit on the PU used subcarriers: in particular, 
the SU can send information over such subcarriers only if 
$\Lpre+1 > \Mvc$ and, thus, $R_{\text{vc}}=\Mvc$. 
In this setting, to allow the SU to transmit as many 
symbols as possible, we assume hereinafter that 
$N = \Lpre-\Mvc+1$, which implies that 
$\Ccal$ is square and nonsingular. 
The design of the matrix $\Ccal$ is discussed in 
Section~\ref{sec:analysis-su}. 

The PU VCs are a precious communication resource for the SU that 
cannot be wasted. For such a reason, the term $\widetilde{\u}_{\SU}(n)$ in the right-hand side (RHS) of \eqref{eq:pre-total}
has been introduced aimed at managing the SU transmission over the
PU VCs. To this goal, we impose that\sfootnote{If
the PU does not use VCs, i.e., $\Mvc=0$, then $\Gcal=\mathbf{O}_{M \times \Mvc}$ and the second summand in the RHS of \eqref{eq:pre-total} disappears.}
$\bm{\gamma}_m=\mathbf{0}_\Mvc$, $\forall m \in \IPUuc$,
which leads to the factorization 
$\Gcal = \bm{\Xi}\, \Dcal$, where the matrix
$\bm{\Xi} \in \setR^{M \times \Mvc}$ 
inserts zero rows in $\Gcal$ over the PU used subcarriers and
$\Dcal \eqdef [\bm{\gamma}_{q_0}, \bm{\gamma}_{q_1}, \ldots, \bm{\gamma}_{q_{\Mvc-1}}]^\herm \in \setC^{\Mvc \times \Mvc}$
is an arbitrary matrix, whose choice is
deferred to Section~\ref{sec:analysis-su},
which is used to 
transmit in parallel a linear combination of the entries of the symbol vector $\xsuvc(n)$ on all the PU VCs.

To limit the average transmit 
power of the STx (in units of energy per PU symbol), we 
consider the frequency-domain version
of the signal \eqref{eq:pre-total} transmitted by the STx, which 
assumes the expression
$\z_2(n) \eqdef \Wdft \, \R_{\text{cp}} \, \widetilde{\z}_2(n) =
\Fxsu(n)  \, \y_2(n) +\Gcal \, \xsuvc(n)$,
where the vector $\y_2(n) \eqdef \Wdft \, \R_{\text{cp}} \, \widetilde{\y}_2(n)=\Hcal_{12} \, \bm{\Theta} \, \x_{\PU}(n) + \v_{2}(n)$
is the frequency-domain block received by the STx [see \eqref{eq:vetr2}].
Power allocation over the different subcarriers is adjusted
at the STx according to the constraint
$\expect\left[\|\z_2(n)\|^2\right]=\PSU$, where
$\PSU >0$ is the SU power budget.
Since
$\expect[|F^{(m)}(n)|^2]=\|\bm{\alpha}_m\|^2$,
$\expect[\|\Gcal \, \xsuvc(n)\|^2]=
\sum_{m \in \IPUvc} \|\bm{\gamma}_{m}\|^2$, and
$\expect[\xsuvc(n) \, \y_2^\herm(n)]=\mathbf{O}_{\Mvc \times M}$,
such a constraint imposes that
\be
(\sigma_{12}^2 \, \PPU + \sigma^2_{v_2}) \sum_{m \in \IPUuc}
\|\bm{\alpha}_m\|^2  + \sum_{m \in \IPUvc} \|\bm{\gamma}_m\|^2 = \PSU 
\label{eq:power}
\ee
where, according to \eqref{eq:C}, one has
$\|\bm{\alpha}_m\|^2=\|\Ccal^\herm \, 
\bm{\Upsilon}_{\text{vc}}^\herm \, 
\bm{\pi}_{\text{IDFT}}^{(m)}\|^2$.

\section{Worst-case ergodic capacity of the PU}
\label{sec:analysis-pu}

Herein, we show that, under
appropriate conditions, the concurrent transmission of
the SU can maintain or even improve the performance of the PU.
With this goal in mind,  we derive the expression of a lower bound on
the mutual information of the PU system with CSI at the receiver (CSIR).
This expression is used to compute a lower bound on the ergodic channel
capacity of the PU, which generalizes and subsumes as a particular case
the results reported in \cite{Verde1,Verde2}.\sfootnote{The upper bound
on the  PU ergodic capacity reported in \cite{Verde1,Verde2}
can be generalized with similar arguments as well.}
Since the detection process at the PRx is
carried out on a frame-by-frame basis, we 
omit the dependence on the frame index $n$ hereinafter. 

With reference to the signal model \eqref{y_PU},
the computation of a general expression of the mutual information
$\MI(\x_{\PU},\y_{\PU}\,|\, \Hcal_{\PU})$ (in bits/s/Hz) between
$\x_{\PU}$ and $\y_{\PU}$, given $\Hcal_{\PU}$,
is significantly complicated by the fact that
$\v_{\PU}$ given by \eqref{eq:eq-noise-3} is not a Gaussian
random vector. However, a lower bound on $\MI(\x_{\PU},\y_{\PU}\,|\, \Hcal_{\PU})$
can be obtained by observing that the ZMCSCG distribution is the
worst-case noise distribution under a variance
constraint.\sfootnote{\label{foot:worst-noise}Given
a  variance  constraint,  the  Gaussian  noise
minimizes  the  capacity
of  a  point-to-point  additive  noise  channel \cite{Gallager_book, CoverThomas_book}, since the Gaussian distribution maximizes
the entropy subject to a variance constraint.}
First of all, to simplify matters, as already done in 
\eqref{sec:SRx}, we replace in \eqref{eq:eq-noise-3} the noise vector 
$\widetilde{\v}_{2}$ with $\hat{\v}_{2} = \Tcp \, \v_2$, 
which allows one to replace 
$\Wdft \, \R_{\text{cp}} \, \widetilde{\H}_{23}^{(0)} \, \widetilde{\F} \, 
\widetilde{\v}_{2}$  with $\Hcal_{23} \, \Fxsu \, \v_{2}$.
Second, by assuming that $\v_{\PU}$ is a ZMCSCG random vector
with (diagonal) correlation matrix
\be
\Rvvpu \eqdef \expect[\v_{\PU} \, \v_{\PU}^\herm] = \sigma^2_{23} \,
(\sigma_{v_2}^2 \, \bm{\Sigma}_{\A} +\bm{\Sigma}_{\G})+ \sigma_{v_3}^2 \, \I_{M}
\label{eq:rvvpu}
\ee
where $\bm{\Sigma}_{\Acal} \eqdef \diag(\|\bm{\alpha}_0\|^2,  \|\bm{\alpha}_1\|^2, \ldots,  \|\bm{\alpha}_{M-1}\|^2)$ and 
$\bm{\Sigma}_{\Gcal} \eqdef \diag(\|\bm{\gamma}_0\|^2,  \|\bm{\gamma}_1\|^2, \ldots,  \|\bm{\gamma}_{M-1}\|^2)$,
and remembering also that
the input distribution maximizing the
capacity of the channel in \eqref{y_PU} is the
ZMCSCG distribution \cite{Gallager_book, CoverThomas_book}, too, that is, $\x_{\PU}$
is a ZMCSCG with correlation matrix
$\expect[\x_{\PU} \, \x_{\PU}^\herm] = \PPU \, \I_Q$,
the conditional mutual information
$\MI(\x_{\PU},\y_{\PU}\,|\, \Hcal_{\PU})$ under an 
average transmit power constraint
is lower bounded \cite{Gallager_book, CoverThomas_book} 
as\sfootnote{The loss in spectral efficiency due to the presence of the CP  
is neglected throughout our capacity analysis.}
\be
\MI(\x_{\PU},\y_{\PU}\,|\, \Hcal_{\PU}) \ge
\frac{1}{M} \, \log_2 \left[ \displaystyle \frac{\det(\Rvvpu + \PPU \, \Hcal_{\PU} \, \bm{\Theta} \, \bm{\Theta}^\trasp \, \Hcal_{\PU}^*)}{\det(\Rvvpu)} \right]
\label{eq:MIbound}
\ee
where $\expect[\|\x_{\PU}\|^2]= Q \, \PPU$ is the transmit power
and the matrix $\Rvvpu$ is nonsingular, i.e., $\det(\Rvvpu) \neq 0$,
for SNR values of practical interest.

The ergodic capacity
is given by $\CPU \eqdef \expect[\MI(\x_{\PU},\y_{\PU}\,|\, \Hcal_{\PU})]$,
where the ensemble average is taken with respect to $\Hcal_{\PU}$.
By virtue of \eqref{eq:MIbound}, it follows that
\barr
\CPU  \ge \CPUlower & \eqdef \frac{1}{M} \, \expect \left\{{\log_2 \left[ \displaystyle \frac{\det(\Rvvpu + \PPU \, \Hcal_{\PU} \, \bm{\Theta} \, \bm{\Theta}^\trasp \, \Hcal_{\PU}^*)}{\det(\Rvvpu)} \right]}\right\}
\nonumber \\ & =
\frac{1}{M} \sum_{m=0}^{M-1} \expect \left\{ \log_2 \left[ 1+
\frac{\PPU \, |H_{\PU}(m)|^2 \, \beta_m}{\sigma^2_{23} \,
(\sigma_{v_2}^2 \, \|\bm{\alpha}_m\|^2  + \|\bm{\gamma}_m\|^2)+ \sigma_{v_3}^2}\right] \right\}
\nonumber \\ & =
\frac{1}{M} \sum_{m \in \IPUuc} \expect \left[ \log_2 \left( 1+
\frac{\PPU \, |H_{\PU}(m)|^2}{\sigma^2_{23} \,
\sigma_{v_2}^2 \, \|\bm{\alpha}_m\|^2 + \sigma_{v_3}^2}\right) \right]
\label{eq:E-capacity}
\earr
where we have remembered that $\beta_m=1$ and $\bm{\gamma}_m=\mathbf{0}_N$ if $m \in \IPUuc$, whereas
$\beta_m=0$ otherwise.
It is noteworthy that equality in \eqref{eq:E-capacity} holds when
the SU is inactive, i.e., $\Acal=\Gcal=\Ob_{M \times N}$: indeed, in this case,
it results that $\v_{\PU} \equiv \v_{3}$ is a ZMCSCG random vector
with correlation matrix
$\Rvvpu =\sigma_{v_3}^2 \, \I_{M}$ and $\CPU$
ends up to the ergodic capacity $\CPUd$ of the direct PU link, given by
\barr
\CPUd & =
\frac{1}{M} \sum_{m \in \IPUuc} \expect \left[ \log_2 \left( 1+
\frac{\PPU \, |H_{13}(m)|^2}{\sigma_{v_3}^2}\right) \right]
\nonumber \\ & = \frac{Q \, \log_2(e)}{M}  \, \Psi(\ASNRd)
\label{eq:Ed-capacity}
\earr
where $\ASNRd \eqdef (\sigma_{13}^2 \, \PPU)/\sigma_{v_3}^2$
is the average SNR at the PRx when $\Acal=\Gcal=\Ob_{M \times N}$,
$\Psi(A)  \eqdef \int_{0}^{+\infty} e^{-u} 
\, \ln(1+ A \, u) \, \mathrm{d}u$,\sfootnote{\label{foot:phi}It is seen \cite{Ozarow} that
\[
\Psi(A) =
- e^{\frac{1}{A}} \, \mathrm{Ei}\left(-\frac{1}
{A}\right)
\approx
\begin{cases}
A  \: ,                & \text{for $0 < A \ll 1$;} \\
\ln(1+A)-\gamma \: ,   & \text{for $A \gg 1$.}
\end{cases}
\]
where, for $x < 0$,
\[
\text{Ei}(x) \eqdef  \int_{-\infty}^{x} \frac{e^{u}}{u} \,
\mathrm{d}u = \gamma + \ln(-x) + \sum_{k=1}^{+\infty}\frac{x^k}{k! \, k}
\label{Ei}
\]
denotes the exponential integral function
and
$
\gamma \eqdef \lim_{n \to \infty}\left( n^{-1} \sum_{k=1}^{n} k^{-1} - \ln n \right)
\approx  0.57721
$
is the Euler-Mascheroni constant.} with $A>0$, 
and we have used the fact that
$|H_{13}(m)|^2$ has an exponential distribution
with mean $\sigma_{13}^2$.

The degree of difficulty in evaluating the expectation in \eqref{eq:E-capacity}
depends on the choice of the precoding matrix $\bm{\Delta}$, which might be optimized
to enhance the performance of the SU system and, hence, may be
a function of the relevant channel coefficients (see Section~\ref{sec:analysis-su}).
To obtain easily interpretable analytical
results,  we assume that $\{\|\bm{\alpha}_m\|^2\}_{m \in \IPUuc}$
and $\{\|\bm{\gamma}_m\|^2\}_{m \in \IPUvc}$
are independent on the realization of the channels, 
which happens, e.g., when a uniform power
allocation strategy is employed
by the STx over its used subcarriers.
In this case,
it is useful to observe from \eqref{eq:HPUentry}
that, conditioned on $H_{23}(m) \, F^{(n)}(m)$, one obtains
that $H_{\PU}(m)$ is a ZMCSCG random variable  with
variance
$\sigma_{13}^2 + \sigma_{12}^2  \, |H_{23}(m)|^2 \, |F^{(n)}(m)|^2$,
$\forall m \in \{0,1,\ldots,M-1\}$,
whose squared magnitude is exponentially distributed
with mean $\sigma_{13}^2 + \sigma_{12}^2  \, |H_{23}(m)|^2 
\, |F^{(n)}(m)|^2$.
By applying the conditional expectation rule \cite{Papoulis}, one has
\be
\CPU \ge \CPUlower =\frac{\log_2(e)}{M} \sum_{m \in \IPUuc}
\expect[\Psi(\ASNRP)]
\label{eq:EE-capacity}
\ee
with
\be
\ASNRP \eqdef
\ASNRd \,
\frac{1 + |H_{23}(m)|^2 \, |\bm{\alpha}^\herm_m \, \xsuuc|^2 \,  \frac{\sigma_{12}^2}{\sigma_{13}^{2}}}
{1+ \|\bm{\alpha}_m\|^2 \, \sigma_{23}^2  \, \frac{\sigma_{v_2}^2}{\sigma_{v_3}^{2}}}
\label{eq:gamma3m}
\ee
where we have remembered that
$F^{(n)}(m) = \bm{\alpha}^\herm_m \, \xsuuc$.
The lower bound \eqref{eq:EE-capacity} boils down to
that reported in \cite{Verde1,Verde2} when
$\Mvc=0$, $\Lpre=0$ or, equivalently, $N=1$, $\bm{\alpha}_m \equiv \alpha_m=1$,
and $\bm{\gamma}_m \equiv \gamma_m=0$,
for each $m \in \{0,1,\ldots,M-1\}$.

The numerator in \eqref{eq:gamma3m} is the gain
(with respect to the direct PU link) due to AF relaying,
whereas its denominator is the performance loss caused
by noise propagation from the STx to the PRx.
As intuitively expected, if the SU does not transmit over all the $Q$ subcarriers
used by the PU ({\em conventional CR scenario}), i.e., $\bm{\alpha}_m=\ob_N$
for each $m \in \IPUuc$,
one has $\ASNRP = \ASNRd$ and, hence, $\CPU =\CPUd$.
In contrast, in our framework, the $m$th subcarrier is simultaneously used by both the PU and the SU, with $m \in \IPUuc$.
By resorting to the law of total expectation  \cite{Papoulis}, it results that
\begin{multline}
\expect[\Psi(\ASNRP)] = \expect[\Psi(\ASNRP) \,|\, \ASNRP \ge \ASNRd] \,
[1-\Prob(\ASNRP < \ASNRd)] + \\
\expect[\Psi(\ASNRP) \,|\, \ASNRP < \ASNRd] \, \Prob(\ASNRP < \ASNRd) \: .
\end{multline}
It is noteworthy that, if $\Prob(\ASNRP < \ASNRd)  \rightarrow 0$, then
\be
\expect[\Psi(\ASNRP)] = \expect[\Psi(\ASNRP) \,|\, \ASNRP \ge \ASNRd] \ge
\Psi(\ASNRd) \: ,
\quad \forall m \in \IPUuc 
\label{eq:iiqq}
\ee
where the inequality comes from the fact that
$\Psi(A)$ is a monotonically
increasing function of $A \ge 0$.
Bearing in mind \eqref{eq:Ed-capacity} and
\eqref{eq:EE-capacity}, inequality \eqref{eq:iiqq}
implies that
$\CPU \ge \CPUd$. Remarkably, in the presence of the concurrent SU transmission,
the capacity of the PU cannot degrade if
$\Prob(\ASNRP < \ASNRd)$ is negligibly small.
Therefore,  we say that the PU system is in {\em outage}
when $\ASNRP < \ASNRd$ and we will
refer to $\PPUout \eqdef \Prob(\ASNRP < \ASNRd)$ as the {\em outage probability}
of the PU system.
Evaluation of $\PPUout$ requires
the calculation of the cumulative distribution function
$p_m(z) \eqdef P(|H_{23}(m)|^2 \, |\bm{\alpha}^\herm_m
\, \xsuuc|^2 \le z)$ of the random variable
$|H_{23}(m)|^2 \, |\bm{\alpha}^\herm_m \, \xsuuc|^2$,
with $z \ge 0$.
To this aim, we remember that $H_{23}(m)$ is
a ZMCSCG random variable with variance $\sigma_{23}^2$
and, hereinafter, we additionally assume that $\xsuuc$
is a ZMCSCG random vector with correlation matrix
$\expect[\xsuuc \, \xsuuc^\herm] = \I_Q$.
Consequently, it results that
$|H_{23}(m)|^2$ and $|\bm{\alpha}^\herm_m \, \xsuuc|^2$
are independent exponential random  variables
with mean $\sigma_{23}^2$ and
$\|\bm{\alpha}_m\|^2$, respectively, which leads to
$p_m(z) \equiv 0$ for $z<0$, whereas (see, e.g., \cite{Papoulis})
\be
p_m(z) = 1-\frac{1}{\sigma_{23}^2}
\int_{0}^{+ \infty} e^{-\left(\frac{x}{\sigma_{23}^2}+
\frac{z}{x \, \|\bm{\alpha}_m\|^2}\right)}
\, {\rm d}x =
1-\frac{2 \,  \sqrt{z}}{\sigma_{23} \, \|\bm{\alpha}_m\|} \,
K_1\left(\frac{2 \,  \sqrt{z}}{\sigma_{23} \, \|\bm{\alpha}_m\|}\right)
\quad
(z \ge 0)
\ee
where $K_{\alpha}(x)$ is the modified Bessel function  of the third kind and order $\alpha$, with $x >0$.\sfootnote{\label{foot:bessel}As by definition (see, e.g., \cite{AbraSteng_book}) 
\[
K_{\alpha}(x) \eqdef \frac{\sqrt{\pi} \, x^\alpha}{2^\alpha \, \Gamma(\alpha+1/2)}
\int_{0}^{+\infty} e^{-x \, t} \, (t^2-1)^{\alpha-1/2} \, {\rm d}t
\]
where $\Gamma(x) \eqdef \int_{0}^{+\infty} t^{x-1} \, e^{-t} \, {\rm d}t$ is the Gamma function. It results that $\Gamma(1/2)=\sqrt{\pi}$ and $\Gamma(3/2)=\sqrt{\pi}/2$. Moreover, 
for any $p>0$ and $q>0$, it results that (see \cite[Eq.~2.3.16.1]{Prud_book})
\[
\int_{0}^{+\infty} x^{\alpha-1} \, e^{-\left(p \, x + \frac{q}{x}\right)} \, {\rm d}x =
2 \left(\frac{q}{p}\right)^{{\alpha}/{2}} K_\alpha(2 \sqrt{p \, q}) \: .
\]
}
Accounting for
\eqref{eq:gamma3m}, it follows that
\be
\PPUout =p_m\left( \|\bm{\alpha}_m\|^2 \, \frac{\sigma_{23}^2 \, \sigma_{13}^2}{\sigma_{12}^2}  \, \frac{\sigma_{v_2}^2}{\sigma_{v_3}^{2}} \right) =
1- 2 \, \frac{\sigma_{13}}{\sigma_{12}}
\, \frac{\sigma_{v_2}}{\sigma_{v_3}} \,
K_1\left( 2 \, \frac{\sigma_{13}}{\sigma_{12}}
\, \frac{\sigma_{v_2}}{\sigma_{v_3}} \right) \: .
\label{eq:prob_out}
\ee
It is noteworthy that the outage probability of the PU system
does not depend on the precoding matrix of the STx
and $\PPUout \equiv \mathsf{Prob}_{\PU,\text{out}}$.
Henceforth, the following mathematical condition
\be
2 \, \frac{\sigma_{13}}{\sigma_{12}}
\, \frac{\sigma_{v_2}}{\sigma_{v_3}} \,
K_1\left( 2 \, \frac{\sigma_{13}}{\sigma_{12}}
\, \frac{\sigma_{v_2}}{\sigma_{v_3}} \right)  \rightarrow 1
\label{eq:cond-1}
\ee
ensures that the outage probability of the PU system tends to zero and,
thus, $\CPU \ge \CPUd$.

In order to find the solution of eq.~\eqref{eq:cond-1}
with respect to $x \equiv 2 \, ({\sigma_{13}}/{\sigma_{12}})
\, ({\sigma_{v_2}}/{\sigma_{v_3}})$,
it is useful to consider the limiting form of the Bessel function $K_1(x)$
for small argument: when $x \rightarrow 0$, it results \cite[Eq.~9.7.2]{AbraSteng_book} that  $K_1(x) \sim 1/x$; therefore, equation
$x \, K_1(x) \rightarrow 1$ is satisfied for $x$ close to zero. This implies that
eq.~\eqref{eq:cond-1} is fulfilled when
\be
\frac{\sigma_{13}}{\sigma_{12}}
\, \frac{\sigma_{v_2}}{\sigma_{v_3}} \rightarrow 0
\label{eq:cond-2}
\ee
that is, in practical terms, when $\sigma_{13}/\sigma_{12}$ is much smaller
than $\sigma_{v_3}/\sigma_{v_2}$. In this case, it is interesting to observe
from \eqref{eq:gamma3m} that $\expect(\ASNRP)$ turns out to be much greater
than  $\ASNRd$.  In other words,
to achieve the performance gain $\CPU \ge \CPUd$,
the favourable effect of AF relaying has to be
predominant {\em on average} with respect to the adverse phenomenon of 
noise propagation. 

Let us specialize condition \eqref{eq:cond-2}
to a case of practical interest. To this end, we assume that: (i)
$\sigma_{i \ell}^2 = d_{i \ell}^{-\eta}$,
where $d_{i \ell}$ is the distance between
nodes $i$ and $\ell$, and
 $\eta$ denotes the path-loss exponent;
(ii) nodes 2 and 3 (approximatively)
have the same noise figure, i.e.,
$\sigma_{v_2}^2 \approx \sigma_{v_3}^2$.
Under these assumptions, condition
\eqref{eq:cond-2} ends up to $d_{12}/d_{13} \rightarrow 0$:
the outage probability of the PU system is vanishingly small
when the distance $d_{13}$ between the PTx and the PRx
is significantly greater than the distance
between the PTx and the STx (see Fig.~\ref{fig:figure_1}).

A final remark is now in order regarding the dependence of
$\CPUlower$ on the power budget $\PSU$ of the SU.
Accounting for \eqref{eq:power}, it follows that, 
$\forall m \in \mathcal{M}$,
\be
\|\bm{\alpha}_m\|^2= \frac{\PSU-\sum_{\ell \in \IPUvc}
 \|\bm{\gamma}_\ell\|^2}{
\sigma_{12}^2 \, \PPU + \sigma^2_{v_2}} - 
\sum_{\shortstack{\footnotesize $\ell \in \IPUuc$ \\
\footnotesize $\ell \neq m$}} \|\bm{\alpha}_\ell\|^2
\label{eq:a}
\ee
The following Lemma unveils the relationship between $\CPUlower$  and $\PSU$:
\begin{lemma}
\label{lem:cpu}
If \eqref{eq:cond-2} holds, then
$\CPUlower$ in \eqref{eq:EE-capacity} is a
monotonically increasing function of $\PSU$.
\end{lemma}
\begin{IEEEproof}
See Appendix~\ref{app:lemma-cpu}.
\end{IEEEproof}
The statement of Lemma~\ref{lem:cpu} is
in contrast with conventional CR approaches \cite{Haykin2005}, for which
concurrent transmission of the SU is allowed only
if its power is subject to a strict constraint.
Such a result directly comes from the fact that the
STx also acts as a relay for the PU system.

\section{Analysis  of the ergodic capacity and precoding optimization of the SU}
\label{sec:analysis-su}

In this section, we investigate the information-theoretic performance
of the SU and also discuss
how the precoding matrices $\Ccal$ and $\Dcal$
can be optimized to enhance the achievable rate of the SU.
Specifically, we assume that the SRx has perfect knowledge of the matrix ${\Hcal}_{\SU}$, which can be estimated via training 
sent by the STx (see Subsection~\ref{sec:SRx}).\sfootnote{This can
be regarded as a {\em worst case}, since in practice the SRx 
might additionally have knowledge of the training symbols of
the PU system, which may be used to estimate the channel 
impulse response over the $1 \to 4$ link, i.e., $\Hcal_{14}$.}

With reference to the signal model
\eqref{y_SU}, the channel output is represented by
the pair $(\y_{\SU}, {\Hcal}_{\SU})$
and, thus, the mutual information
between channel input and output is given by
$\MI(\x_{\SU},\y_{\SU}\,|\, {\Hcal}_{\SU})$ (in bits/s/Hz).
First, we calculate a lower bound on
$\MI(\x_{\SU},\y_{\SU}\,|\, \Hcal_{\SU})$,
by considering the worst-case distribution for the
equivalent noise term at the SRx under a variance constraint
(see footnote~\ref{foot:worst-noise}),  i.e., $\v_{\SU}$ is modeled as a
ZMCSCG random vector with (diagonal) correlation matrix
$\Rvvsu \eqdef \expect[\v_{\SU} \, \v_{\SU}^\herm] = \PPU \, \sigma^2_{14} \, \bm \Theta \, \bm \Theta^{\trasp} +
\sigma^2_{v_4} \, \I_M$.
By assuming that $\x_{\SU}$ is a ZMCSCG random vector, with correlation
matrix $\expect[\x_{\SU} \, \x_{\SU}^\herm]=\I_N$,
it follows that $\MI(\x_{\SU},\y_{\SU}\,|\, {\Hcal}_{\SU})$ under an average transmitter
constraint is lower bounded as
\be
\MI(\x_{\SU},\y_{\SU}\,|\, {\Hcal}_{\SU})  \geq
\MI_{\text{min}}(\x_{\SU},\y_{\SU}\,|\, {\Hcal}_{\SU})  \eqdef
\frac{1}{M} \,
\log_2 \left[ \displaystyle \frac{\det(\Rvvsu + {\Hcal}_{\SU} \,
\bOmega \, {\Hcal}_{\SU}^\herm)}{\det(\Rvvsu)} \right]
\label{eq:MIboundSU}
\ee
where $\expect[\|\x_{\SU}\|^2]=  N+\Mvc$ is the overall transmit power and $\Rvvsu$ is nonsingular, i.e., $\det(\Rvvsu) \neq 0$,
for SNR values of practical interest,
$\bOmega \eqdef \bm{\Delta} \, \bm{\Delta}^\herm \in \setC^{2 M \times 2 M}$
is a positive-semidefinite Hermitian matrix.\sfootnote{The set of
positive-semidefinite matrices is a closed convex cone \cite{Horn}.}
It is noteworthy that the RHS of \eqref{eq:MIboundSU}
is concave as a function of $\bOmega$ \cite[Thm.~1]{Cover}
and, therefore, it can be maximized with respect to $\bOmega$.
To this aim, using the facts that $\det(\B_1 \, \B_2)=\det(\B_1) \, \det(\B_2)$ and
$\det(\B_1^{-1})=1/\det(\B_1)$, for $\B_1, \B_2 \in \setC^{n \times n}$,
it is readily seen that
\be
\frac{\det(\Rvvsu + {\Hcal}_{\SU} \,
\bOmega \, {\Hcal}_{\SU}^\herm)}{\det(\Rvvsu)} =
\det(\I_M+ \Rvvsu^{-1} \, {\Hcal}_{\SU} \,
\bOmega \, {\Hcal}_{\SU}^\herm) \: .
\label{eq:det}
\ee
By observing that $\Rvvsu$ is diagonal by
construction, Hadamard's inequality \cite{Horn}
implies that the RHS of \eqref{eq:det}
[and, hence, $\MI_{\text{min}}(\x_{\SU},\y_{\SU}\,|\, {\Hcal}_{\SU})$] is maximized when 
\be
{\Hcal}_{\SU} \,
\bOmega \, {\Hcal}_{\SU}^\herm= \H_{\SU} \, \bm{\Pi}_{\text{IDFT}} \, \bm{\Upsilon}_{\text{vc}} \, \Ccal \,  \Ccal^\herm \, \bm{\Upsilon}_{\text{vc}}^\herm \, 
\bm{\Pi}_{\text{IDFT}}^\herm \, {\H}_{\SU}^\herm+
\Hcal_{24} \, \bm{\Xi} \, \Dcal \, \Dcal^\herm \, \bm{\Xi}^\trasp \,  \Hcal_{24}^\herm
\label{eq:diagg}
\ee
is a diagonal matrix. Since $\H_{\SU}$ [see eq.~\eqref{eq:diag-hsu}] and $\Hcal_{24}$ are diagonal matrices,
maximization of $\MI_{\text{min}}(\x_{\SU},\y_{\SU}\,|\, {\Hcal}_{\SU})$ can be obtained by imposing that the matrices $ \bm{\Pi}_{\text{IDFT}} \, \bm{\Upsilon}_{\text{vc}} \, \Ccal \,  \Ccal^\herm \, \bm{\Upsilon}_{\text{vc}}^\herm \, 
\bm{\Pi}_{\text{IDFT}}^\herm$ and $\bm{\Xi} \, \Dcal \, \Dcal^\herm \, \bm{\Xi}^\trasp$ are diagonal, too. Therefore, we impose that 
$ \bm{\Pi}_{\text{IDFT}} \, \bm{\Upsilon}_{\text{vc}} \, \Ccal \,  \Ccal^\herm \, \bm{\Upsilon}_{\text{vc}}^\herm \, 
\bm{\Pi}_{\text{IDFT}}^\herm= \bm{\Sigma}_{\Acal}$
and $\bm{\Xi} \, \Dcal \, \Dcal^\herm \, \bm{\Xi}^\trasp=\bm{\Sigma}_{\Gcal}$,
where $\bm{\Sigma}_{\Acal}$ and $\bm{\Sigma}_{\Gcal}$ have been previously defined in \eqref{eq:rvvpu}, whose particular solutions can be expressed as 
\barr
\Ccal \,  \Ccal^\herm & = \bm{\Upsilon}_{\text{vc}}^\herm \, 
\bm{\Pi}_{\text{IDFT}}^\herm \, \bm{\Sigma}_{\Acal} \, 
\bm{\Pi}_{\text{IDFT}} \, \bm{\Upsilon}_{\text{vc}}
\label{eq:cc}
\\
\Dcal \, \Dcal^\herm  & = \bm{\Xi}^\trasp \, \bm{\Sigma}_{\Gcal} 
\, \bm{\Xi} \:.
\label{eq:dd}
\earr
In this case, by virtue of \eqref{eq:MIboundSU},
\eqref{eq:diagg}, \eqref{eq:cc}, 	and \eqref{eq:dd}, 
and remembering that we have imposed  
$\bm{\alpha}_m=\mathbf{0}_N$, $\forall m \in \IPUvc$,
and 
$\bm{\gamma}_m=\mathbf{0}_\Mvc$, $\forall m \in \IPUuc$, one has
\begin{multline}
\MI_{\text{min}}(\x_{\SU},\y_{\SU}\,|\, {\Hcal}_{\SU}) =
\frac{1}{M} \, \left [
\sum_{m \in \IPUuc} \log_2 \left( 1 + \frac{|H_{\SU}(m)|^2 \,\|\bm{\alpha}_m\|^2}{ \sigma^2_{14} \, \PPU  + \sigma_{v_4}^2} \right) 
\right. \\ \left. +
\sum_{m \in \IPUvc} \log_2 \left( 1 +  \frac{|H_{24}(m)|^2 \, \|\bm{\gamma}_m\|^2}{\sigma_{v_4}^2} \right) \right ] \:.
\label{eq:I-max}
\end{multline}
By averaging $\MI(\x_{\SU},\y_{\SU}\,|\, \Hcal_{\SU})$ with respect to
the relevant channel parameters, and relying on
\eqref{eq:MIboundSU} and \eqref{eq:I-max}, the  ergodic capacity $\CSU$
of the SU can be lower bounded as follows
\be
\CSU \eqdef \expect[\MI(\x_{\SU},\y_{\SU}\,|\, \Hcal_{\SU})]  
\ge \CSUlower \eqdef
\expect \left[ \MI_{\text{min}}(\x_{\SU},\y_{\SU}\,|\, \Hcal_{\SU}) \right]  \: .
\label{eq:E-capacity-SU}
\ee
It is worth noticing that the capacity of the SU is essentially
limited by the variance $\sigma^2_{14} \, \PPU  + \sigma_{v_4}^2$
of the equivalent noise term $\v_{\SU}$ at the SRx
(see~Subsection~\ref{sec:SRx}).
Evaluation of the expectation in \eqref{eq:E-capacity-SU}
depends on the choice of the scalar variables
$a_m \eqdef \|\bm{\alpha}_m\|^2$, for $m \in \IPUuc$,
and $g_m \eqdef \|\bm{\gamma}_m\|^2$, for $m \in \IPUvc$,
which in its turn may depend on the CSI at the transmitter (CSIT)
of the SU system. In the following two subsections, we
separately consider two relevant scenarios.

\subsection{CSIT scenario}
\label{sec:csit}

In this scenario, the STx has perfect knowledge of the channel matrix
$\Hcal_{\SU}$, which allows one to  further maximize
the mutual information between channel input and output.
Channel estimation at the transmitter requires either
a feedback channel or the application of the channel reciprocity
property when the same carrier frequency is used for transmission and
reception. Henceforth, accounting for \eqref{eq:I-max},
we propose to solve the following optimization problem
\be
\arg \max_{\shortstack{\footnotesize $\{a_m\}_{m \in \IPUuc}$
\\ \footnotesize $\{g_m\}_{m \in \IPUvc}$}}
\left[ \sum_{m \in \IPUuc} \log_2 \left(1 +
\frac{|H_{\SU}(m)|^2 \, a_m}{\sigma^2_{14} \, \PPU + \sigma_{v_4}^2} \right)
+ \sum_{m \in \IPUvc} \log_2 \left( 1 +  \frac{|H_{24}(m)|^2 \, g_m}{\sigma_{v_4}^2} \right) \right ]
\label{eq:max-min-2}
\ee
subject to  the power constraint [see \eqref{eq:power}]
\be
(\sigma_{12}^2 \, \PPU + \sigma^2_{v_2}) \sum_{m \in \IPUuc}
a_m + \sum_{m \in \IPUvc} g_m = \PSU \: .
\label{eq:power-2}
\ee
The solution of such a problem is given by the following Lemma:
\begin{lemma}
\label{lem:water}
Problem \eqref{eq:max-min-2}--\eqref{eq:power-2}
admits the following {\em waterfilling} solution
\barr
a_{m, \text{opt}} &= \left[\mu - \frac{\sigma^2_{14} \, \PPU+ \sigma_{v_4}^2}
{|H_{\SU}(m)|^2}\right]^{+}
\:, \quad \forall m \in \IPUuc
\label{eq:waterfilling-a}
\\
g_{m, \text{opt}} &= \left[\mu - \frac{\sigma_{v_4}^2}
{|H_{24}(m)|^2}\right]^{+}
\:, \quad \forall m \in \IPUvc
\label{eq:waterfilling-g}
\earr
where the constant $\mu$ is chosen so as to fulfil
the constraint
\be
(\sigma_{12}^2 \, \PPU + \sigma^2_{v_2}) \sum_{m \in \IPUuc}
\left[\mu - \frac{\sigma^2_{14} \, \PPU+ \sigma_{v_4}^2}
{|H_{\SU}(m)|^2}\right]^{+} +
\sum_{m \in \IPUvc} \left[\mu - \frac{\sigma_{v_4}^2}
{|H_{24}(m)|^2}\right]^{+} = \PSU \: .
\label{eq:mu}
\ee
\end{lemma}
\begin{IEEEproof}
The proof is obtained by using standard optimization concepts
(see, e.g., \cite{Tse_book}).
\end{IEEEproof}
In such a CSIT scenario,
the worst-case ergodic channel capacity of the SU can be
obtained by replacing in \eqref{eq:I-max}--\eqref{eq:E-capacity-SU}
$\|\bm{\alpha}_m\|^2$ and $\|\bm{\gamma}_m\|^2$
with $a_{m, \text{opt}}$ and $g_{m, \text{opt}}$, respectively, thus obtaining
\begin{multline}
\CSUloweropt = \frac{1}{M} \, \left \{
\sum_{m \in \IPUuc}\expect \left[ \left(\log_2 \left(
\frac{\mu \, |H_{\SU}(m)|^2}{\sigma^2_{14} \, \PPU + \sigma_{v_4}^2} \right) \right)^{+} \right] 
\right. \\ \left. +
\sum_{m \in \IPUvc} \expect \left[ \left(\log_2 \left( \frac{\mu \, |H_{24}(m)|^2}{\sigma_{v_4}^2} \right) \right)^+ \right ] \right \} \: .
\label{eq:CCSSU}
\end{multline}
Let us assume for simplicity that nodes 2 and 4 (approximatively)
have the same noise figure, i.e.,
$\sigma^2 = \sigma_{v_2}^2 \approx \sigma_{v_4}^2$, it is readily
verified that, even in the presence of CSIT, the first summand of
the ergodic channel capacity $\CSUloweropt$ tends to a
bounded quantity as $\sigma^2 \rightarrow 0$.
In other words, maximization of the
mutual information between channel input and output
does not allow to cope with the interference generated
by the PU on the SU system over the $1 \rightarrow 4$ link.

\subsection{No CSIT (NOCSIT) scenario}
\label{sec:nocsit}

In this scenario, CSIT is not available at the STx.
In such a case, a viable choice consists of uniformly allocating the
power over the subcarriers used by the SU,
i.e., $\|\bm{\alpha}_m\|^2 \equiv a > 0$
for each $m \in \IPUuc$ and $\|\bm{\gamma}_m\|^2 \equiv g > 0$ for
each $m \in \IPUvc$. In this case, eq.~\eqref{eq:a} ends up to
\be
a= \frac{\PSU-\Mvc \, g}{Q \, 
(\sigma_{12}^2 \, \PPU + \sigma^2_{v_2})} \: .
\label{eq:aa}
\ee
Accounting for \eqref{eq:diag-hsu} and \eqref{eq:E-capacity-SU},
the expectation of the first summand of $\MI_{\text{min}}(\x_{\SU},\y_{\SU}\,|\, {\Hcal}_{\SU})$ in \eqref{eq:I-max} can be evaluated by
exploiting the statistical independence
between $H_{12}(m)$ and $v_2(m)$: indeed, conditioned on $H_{24}(m)$ and $x_{\PU}(m)$, $H_{\SU}(m)$ is a
ZMCSCG random variable having variance $|H_{24}(m)|^2
\, (\sigma^2_{12} \, |x_{\PU}(m)|^2 + \sigma^2_{v_2})$, for each
$m \in \IPUuc$, whose squared magnitude is exponentially distributed with mean $|H_{24}(m)|^2 \, (\sigma^2_{12} \, |x_{\PU}(m)|^2 + \sigma^2_{v_2})$.
Therefore, it results from \eqref{eq:I-max}--\eqref{eq:E-capacity-SU} that
\be
\CSUlowerunif =
\frac{\log_2(e)}{M} \, \left \{
\sum_{m \in \IPUuc} \expect[\Psi(\ASNRS)]
+ \Mvc \, \Psi(\ASNRdd) \right \}
\label{eq:EE-capacity-SU-UP}
\ee
with
\be
\ASNRS \eqdef
\frac{|H_{24}(m)|^2 \, (\sigma^2_{12} \, |x_{\PU}(m)|^2  + \sigma^2_{v_2}) \,
(\PSU-\Mvc \, g)}{Q\, (\sigma_{14}^2 \, \PPU + \sigma^2_{v_4}) \, (\sigma_{12}^2 \, \PPU + \sigma^2_{v_2})}
\label{eq:gamma4m}
\ee
and $\ASNRdd \eqdef (\sigma_{24}^2 \, g)/\sigma_{v_4}^2$ representing the average SNR  of the direct link between STx and SRx,
where we have accounted for \eqref{eq:a} and,
regarding the second summand of $\MI_{\text{min}}(\x_{\SU},\y_{\SU}\,|\, \widetilde{\H}_{\SU})$ in \eqref{eq:I-max}, we have used the fact that
$|H_{24}(m)|^2$ has an exponential distribution
with mean $\sigma_{24}^2$.

A particularization of \eqref{eq:EE-capacity-SU-UP} can be obtained
by assuming that the PU symbols are drawn from a constant-modulus
constellation, i.e., $|x_{\PU}(m)|^2=\PPU$. In this case, when  $\ASNRS$
assumes negligible values  on average, i.e.,
$\PSU-\Mvc \, g \ll Q \, (\sigma_{14}^2 \, \PPU + \sigma^2_{v_4})/\sigma_{24}^2$, one obtains  that
(see footnote~\ref{foot:phi})
\barr
\CSUlowerunif & \approx
\frac{\log_2(e)}{M} \, \left [
\sum_{m \in \IPUuc} \expect(\ASNRS)
+ \Mvc \, \Psi(\ASNRdd) \right ]
\nonumber \\ &=
\frac{\log_2(e)}{M} \left[ \ASNRdd \, \frac{\frac{\PSU}{g}-\Mvc}
{1 + \ASNRddd}+ \Mvc \, \Psi(\ASNRdd)\right]
\label{eq:particular-1}
\earr
with $\ASNRddd \eqdef (\sigma_{14}^2 \, \PPU)/\sigma_{v_4}^2$ representing the average SNR  of the direct link between PTx and SRx.
On the contrary, when  $\ASNRS$
assumes large values on average, i.e.,
$\PSU-\Mvc \, g \gg Q \, (\sigma_{14}^2 \, \PPU + \sigma^2_{v_4})/\sigma_{24}^2$, one gets 
that (see footnote~\ref{foot:phi} again)
\barr
\CSUlowerunif & \approx
\frac{\log_2(e)}{M} \, \left \{
\sum_{m \in \IPUuc} \expect[(\ln(1+\ASNRS) - \gamma]
+ \Mvc \, \Psi(\ASNRdd) \right \}
\nonumber \\ &=
\frac{\log_2(e)}{M} \left[ \Psi \left(\frac{\ASNRdd}{Q} \, \frac{\frac{\PSU}{g}-\Mvc}
{1 + \ASNRddd}\right) -\gamma \, Q + \Mvc \, \Psi(\ASNRdd)\right] \: .
\label{eq:particular-2}
\earr
As it is apparent from \eqref{eq:particular-1} and \eqref{eq:particular-2},
due to the equivalent noise term $\v_{\SU}$ at the SRx
(see~Subsection~\ref{sec:SRx}), the
worst-case capacity of the SU is inversely related to the average SNR over
the direct link between the PTx and the SRx, which might be
a limiting factor for the SU ergodic capacity.
Such a potential trouble can be circumvented by allowing 
the SRx to estimate the PU symbol block $\x_{\PU}$
and, consequently,  subtract its contribution from the 
received data. This requires knowledge at the SRx 
of the training protocol of the PU. 

\section{Numerical performance analysis}
\label{sec:simul}

To corroborate our information-theoretic analysis, we report
some results of numerical simulations.
With reference to
Fig.~\ref{fig:figure_1}, we normalize the
distance between the PTx and the PRx,
as well as the transmitting power of the PU,
by setting $d_{13}=1$ and $\PPU=1$, respectively. 
Specifically,  the nodes $1$ (PTx), $3$ (PRx), 
and $4$ (SRx) have coordinates equal 
to $(-0.5,0)$, $(0.5,0)$,  and $(0,2)$, respectively.
In all the plots where the distance $d_{12}$ varies,  the node $2$ 
(STx) moves along the line  joining the nodes $1$ and $2$,
with $\vartheta=\pi/3$ (see Fig.~\ref{fig:figure_1}).
The memory of the discrete-time channels among the nodes is set equal to
$L_{12}=1$, $L_{13} =L_{14}=3$, and $L_{24} =L_{23}=2$, whereas the corresponding time offsets are fixed to $\theta_{12}=1$, 
$\theta_{13}=\theta_{14}=3$, and  $\theta_{24}=\theta_{23}=2$, 
respectively. The path-loss exponent is chosen equal to $\eta=3$. 
According to \eqref{eq:CP} and \eqref{eq:CP-2},
we choose $\Lpre=10$, which leads to 
$N = \Lpre-\Mvc+1=7$.
The symbol blocks $\x_{\PU}$ and
$\x_{\SU}$ are ZMCSCG random vectors, with correlation
matrices $\PPU \, \I_Q$ and
$\PSU \, \I_{N+\Mvc}$, respectively.
Moreover, we set
$\sigma_{v_2}^{2} = \sigma_{v_3}^{2}=\sigma_{v_4}^{2}=\sigma^{2}$.
The ensemble averages (with respect to the fading channels and information-bearing symbols)
in \eqref{eq:EE-capacity}, \eqref{eq:CCSSU}, and \eqref{eq:EE-capacity-SU-UP}
%and the outage probability $P(|H_{23}(m)|^2 \, |\bm{\alpha}^\herm_m \, \xsuuc|^2 \le z)$ 
are evaluated  through $10^6$ Monte Carlo trials.

\begin{figure*}[t!]
\centering
\begin{minipage}[t]{0.45\textwidth}
\includegraphics[width=\linewidth, trim=20 20 40 20]{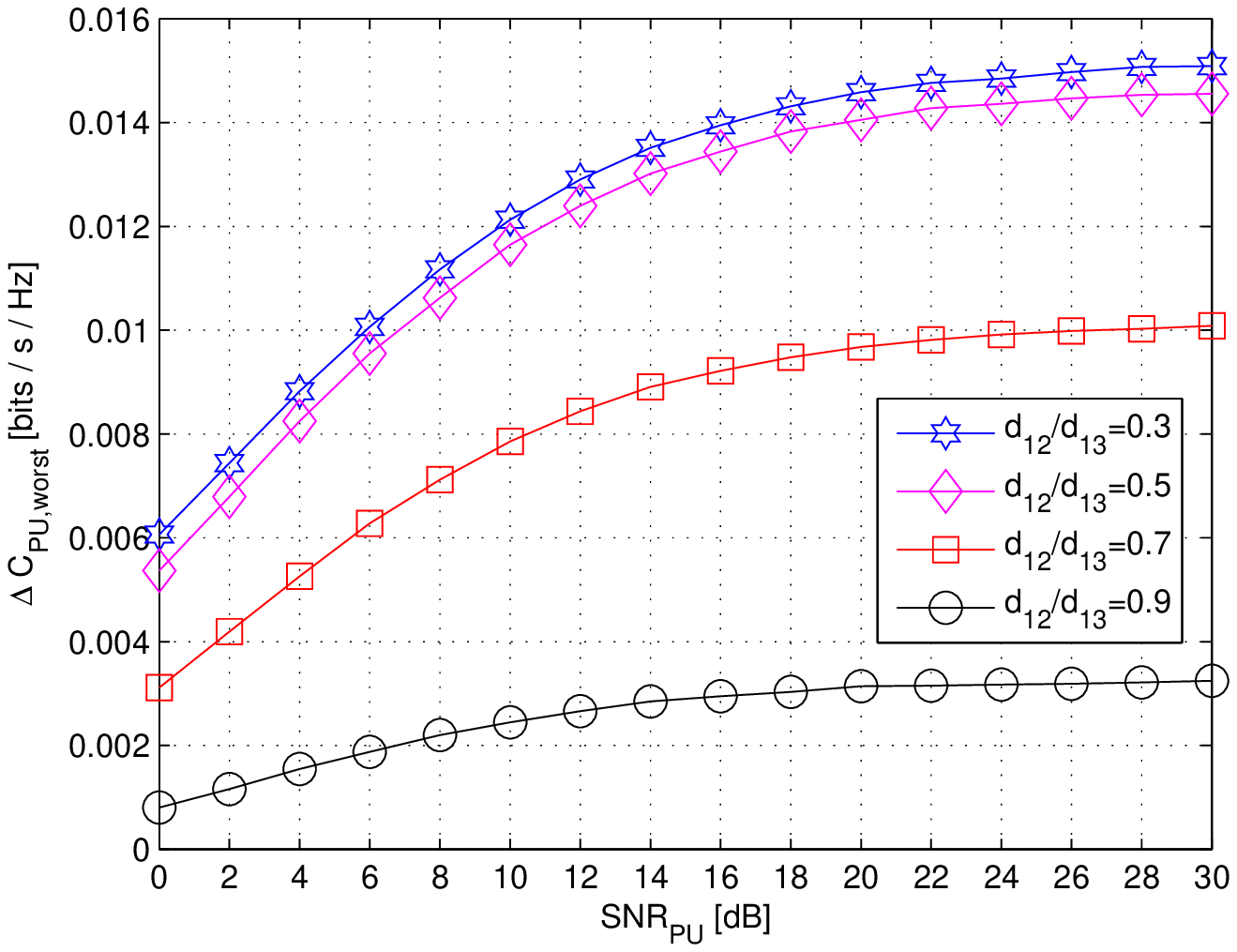}
\caption{$\Delta\CPUworst$ versus $\SNR_{\PU}$ for different values of $d_{12}/d_{13}$ ($\PSU/\PPU=1$).}
\label{fig:fig_3}
\end{minipage}%
\hspace{0.04\textwidth}%
\begin{minipage}[t]{0.45\textwidth}
\includegraphics[width=\linewidth, trim=20 20 40 20]{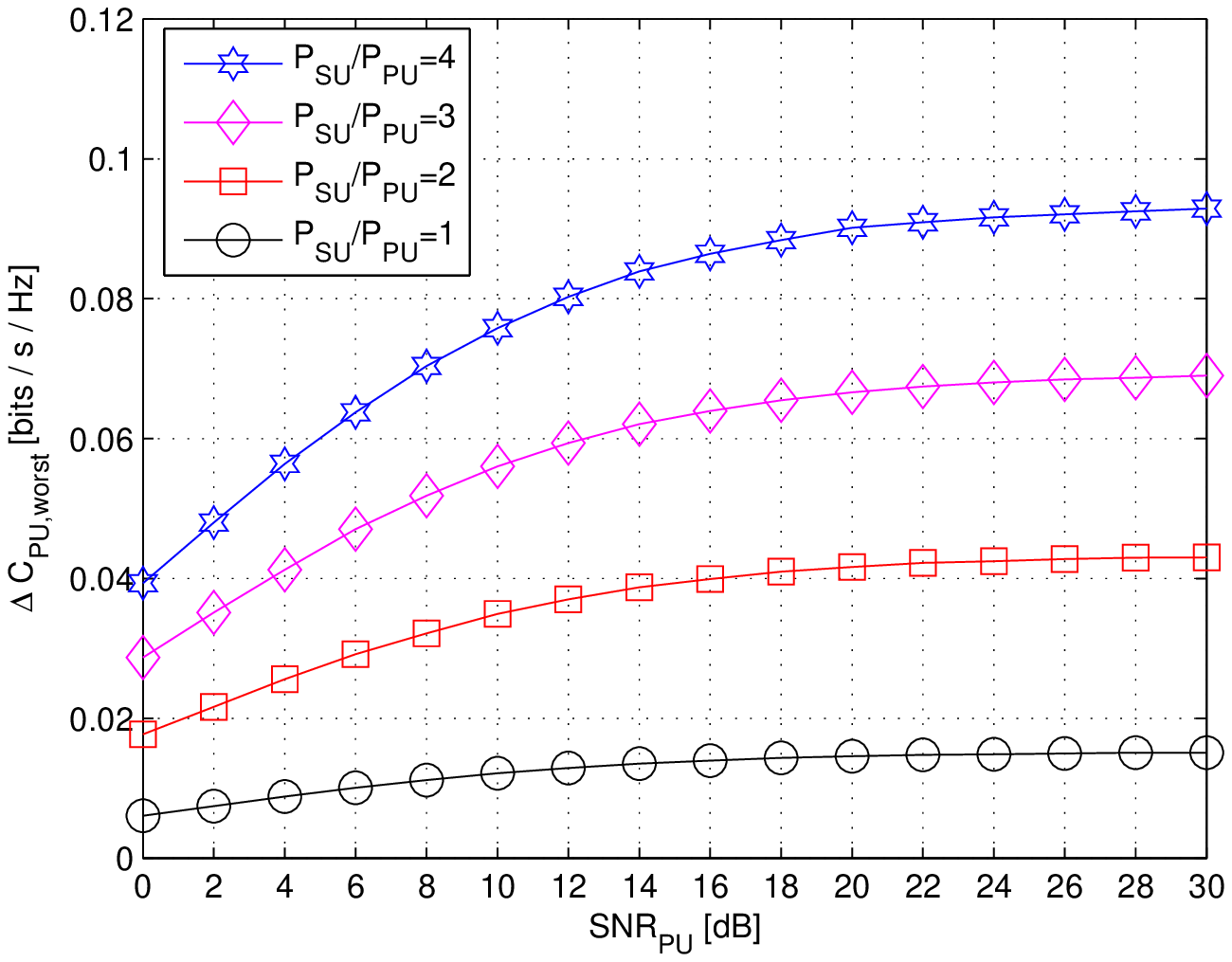}
\caption{$\Delta\CPUworst$ versus $\SNR_{\PU}$ for different values of
$\PSU/\PPU$ ($d_{12}/d_{13}=0.3$).}
\label{fig:fig_4}
\end{minipage}
\end{figure*}

\subsection{Performance of the primary system}

Herein, we study the worst-case performance of the primary system,
by assuming a uniform power allocation for the SU 
transmission,\sfootnote{Results non reported here show that the performance
of the PU is not significantly influenced on the way the SU 
encodes its information-bearing symbols.} i.e., 
$\|\bm{\alpha}_m\|^2 \equiv a > 0$,
for each $m \in \IPUuc$, and $\|\bm{\gamma}_m\|^2 \equiv g > 0$, for
each $m \in \IPUvc$, fulfil  \eqref{eq:aa}, with 
$\PSU/\PPU=1$ and $g=\PSU/(2 \, \Mvc)$.

Figs.~\ref{fig:fig_3} and \ref{fig:fig_4} depict 
the (minimum) capacity gain
$\Delta\CPUworst \eqdef \CPUlower-\CPUd$ of the PU
as a function of $\SNR_{\PU} \eqdef \PPU/\sigma^{2}$.
Specifically, different values of the ratio $d_{12}/d_{13}$
are considered in Fig.~\ref{fig:fig_3}, with $\PSU/\PPU = 1$,
whereas the curves in Fig.~\ref{fig:fig_4} are reported
for different values of the power ratio $\PSU/\PPU$,
with $d_{12}/d_{13}=0.3$.
Results show that the PU can unknowingly attain a capacity gain from the concurrent transmission of the SU, which
significantly increases either when the SU is getting closer and closer to the PU
or when the SU system has a power budget to spend greater than that of the PU one. For instance, let us consider the case of a primary Wi-Fi system with 
$1/T_c=20$ MHz: when $\PPU=\PSU$
and $d_{12}/d_{13}=0.3$, it results
from Fig.~\ref{fig:fig_3} that the capacity gain is at least equal to $300$ kbps
at $\SNR_{\PU} > 20$ dB, whereas, when the STx spends 
twice as much power as the PTx, such a gain amounts at least to $1.8$ Mbps
(see Fig.~\ref{fig:fig_4}). 

\begin{figure*}[t!]
\centering
\includegraphics[width=0.95\linewidth, trim=20 40 40 20]{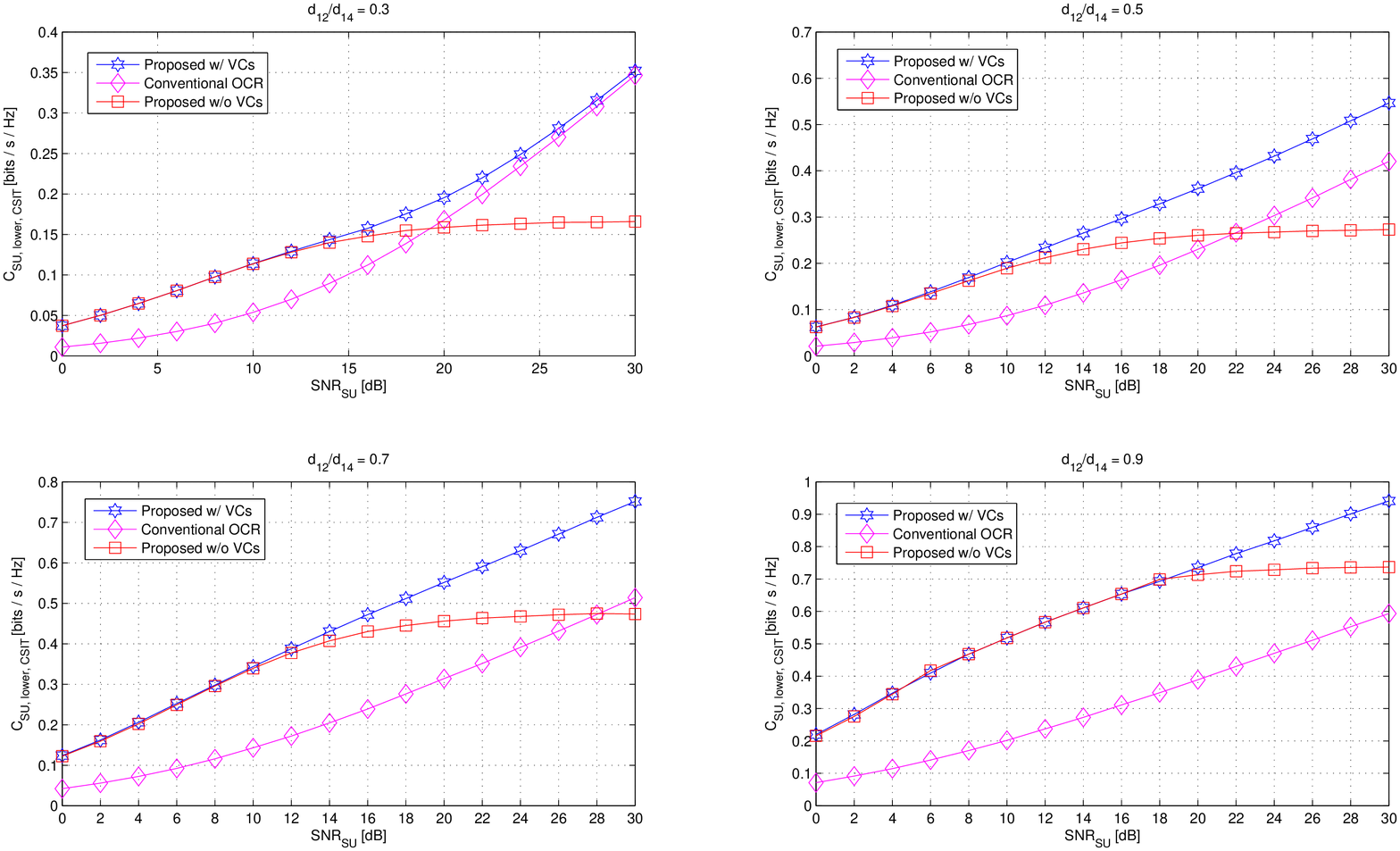}
\caption{$\CSUloweropt$ versus $\SNR_{\SU}$ for 
different values of $d_{12}/d_{14}$ ($\PSU/\PPU=1$
and CSI at the STx).}
\label{fig:fig_6}
\end{figure*}
\begin{figure*}[t!]
\centering
\includegraphics[width=0.95\linewidth, trim=20 40 40 20]{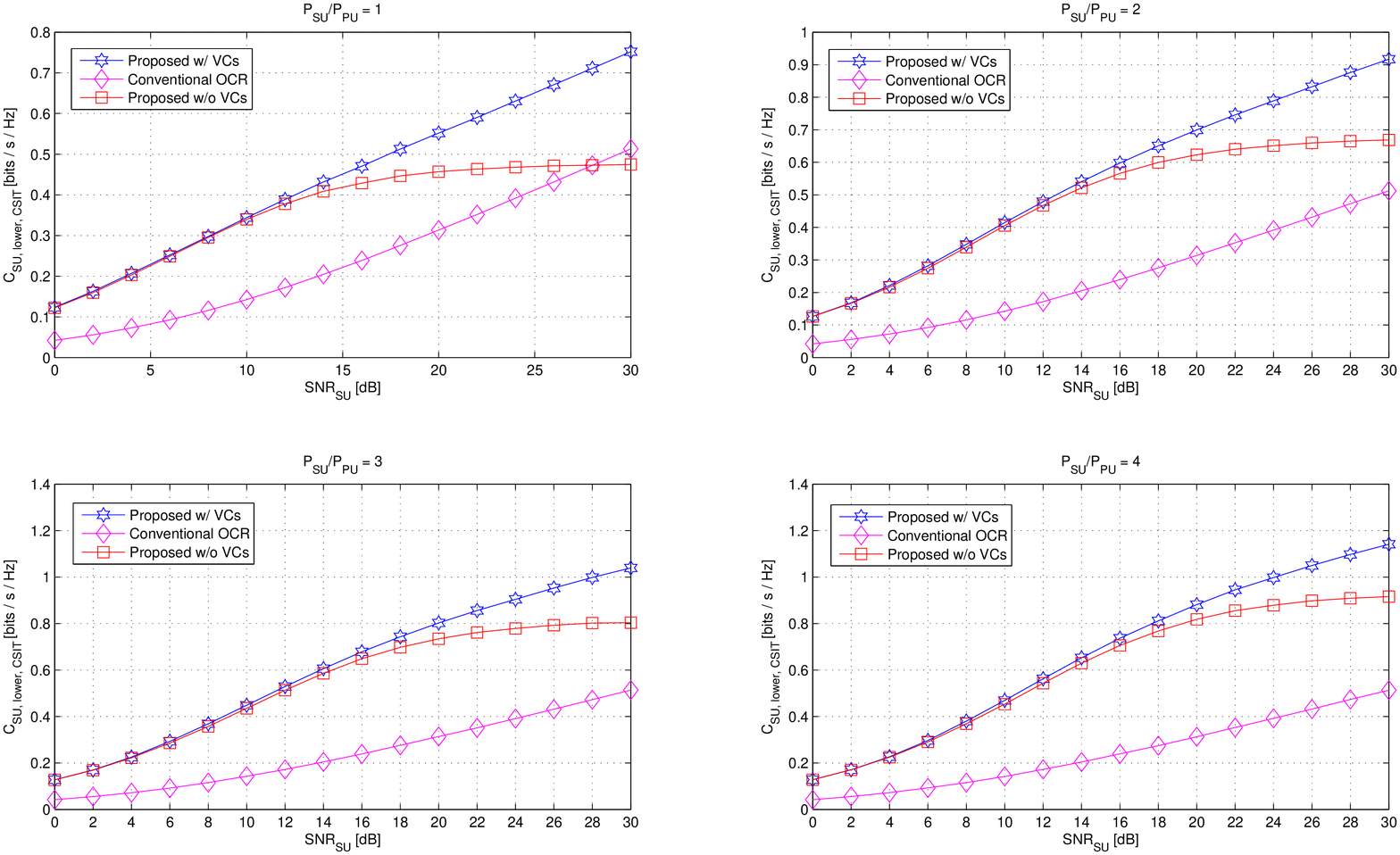}
\caption{$\CSUloweropt$ versus $\SNR_{\SU}$ for 
different values of $\PSU/\PPU$ ($d_{12}/d_{14}=0.7$
and CSI at the STx).}
\label{fig:fig_7}
\end{figure*}

\subsection{Performance of the secondary system}

In this subsection, we focus on  the (minimum) achievable rate $\CSUlower$
of the SU [see \eqref{eq:E-capacity-SU}].
In particular, we consider the case when CSI is available
at the STx by depicting $\CSUloweropt$ in 
Figs.~\ref{fig:fig_6} and \ref{fig:fig_7}, as well 
as the case in which the STx has no CSI 
by reporting $\CSUlowerunif$ in 
Figs.~\ref{fig:fig_8} and \ref{fig:fig_9}.
In both cases, we compare two different implementations
of the proposed method: in the former one, referred to 
as ``Proposed w/ VCs'', according to \eqref{eq:pre-total},
the SU transmits on both the used 
and virtual subcarriers of the PU; in the latter one, 
referred to as ``Proposed w/o VCs'', 
the SU sends its symbols only over 
the used subcarriers of the PU, i.e., 
$\Gcal=\mathbf{O}_{M \times \Mvc}$ in \eqref{eq:pre-total}.
Additionally, as a performance comparison, we report 
the exact ergodic capacity $\CPUd$ of the OCR scheme,
referred to  as ``Conventional OCR'', 
when the SU transmits only on the VCs of the PU,
i.e., $\widetilde{\F}(n)=\mathbf{O}_{P \times P}$
in \eqref{eq:pre-total};
we also plot the
worst-case capacity of the NORC scheme 
\cite{Verde2} when no CSI is available, referred to  as ``NOCR \cite{Verde2}'',
which is obtained from
\eqref{eq:pre-total} by setting $\Lpre=0 \Rightarrow N=1$, $\Gcal=\mathbf{O}_{M \times \Mvc}$, 
and $\Acal= \sqrt{a} \, [1, \ldots, 1]^\trasp \in \setR^M$,
where $a$ is given by \eqref{eq:aa} with $g=0$.

\begin{figure*}[t!]
\centering
\includegraphics[width=0.95\linewidth, trim=20 40 40 20]{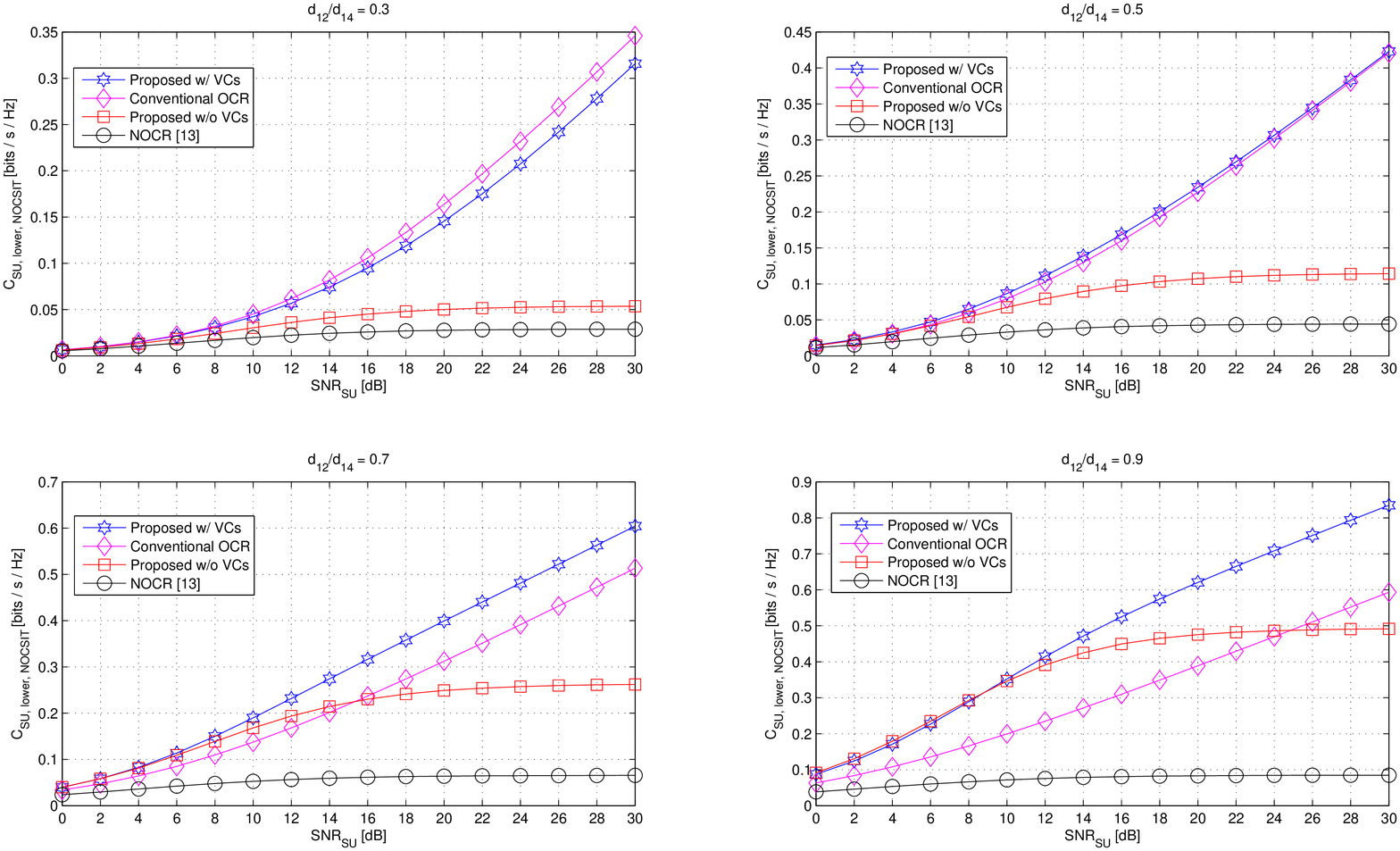}
\caption{$\CSUlowerunif$ versus $\SNR_{\SU}$ for 
different values of $d_{12}/d_{14}$ ($\PSU/\PPU=1$
and no CSI at the STx).}
\label{fig:fig_8}
\end{figure*}
\begin{figure*}[t!]
\centering
\includegraphics[width=0.95\linewidth, trim=20 40 40 20]{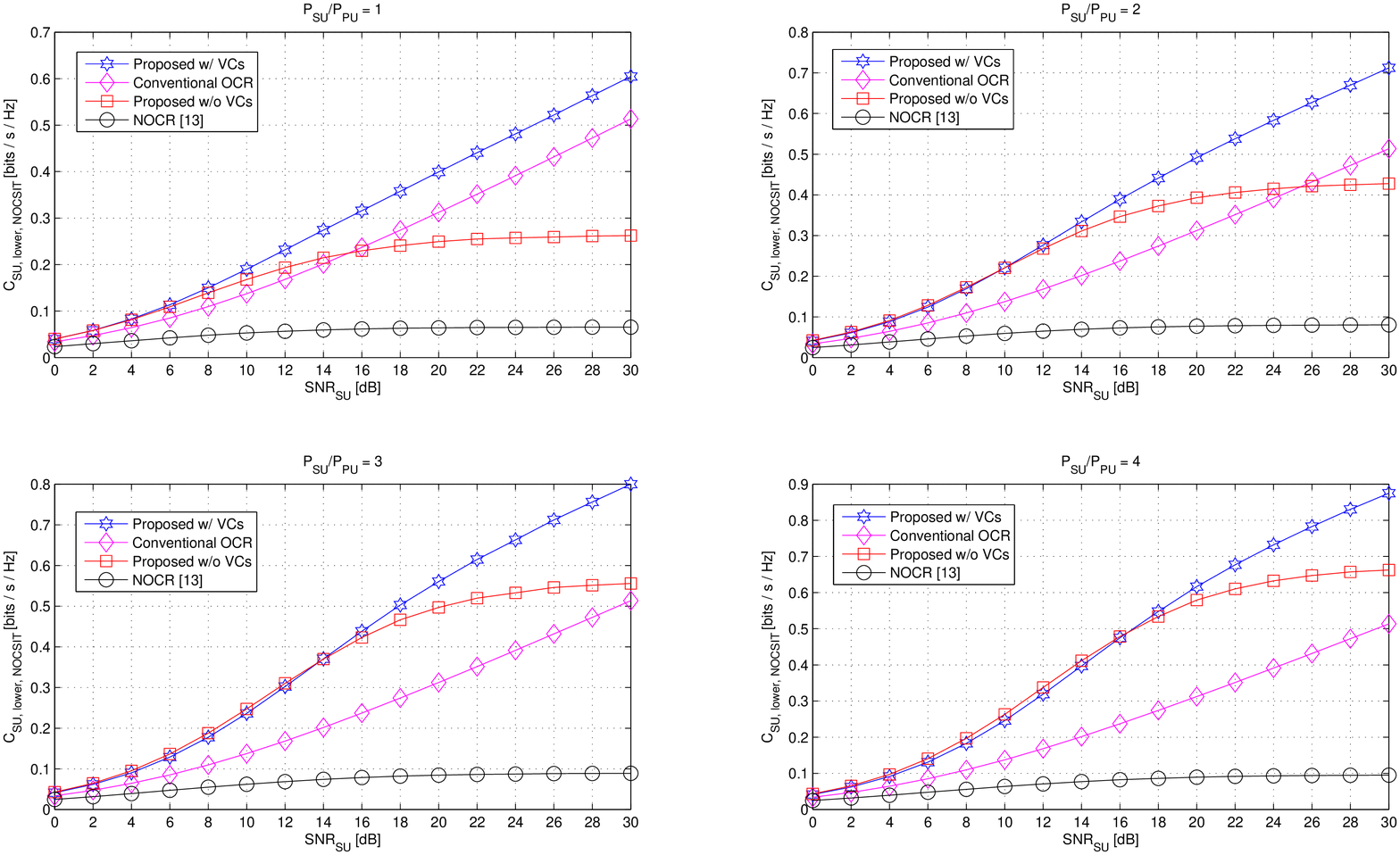}
\caption{$\CSUlowerunif$ versus $\SNR_{\SU}$ for 
different values of $\PSU/\PPU$ ($d_{12}/d_{14}=0.7$
and no CSI at the STx).}
\label{fig:fig_9}
\end{figure*}

Figs.~\ref{fig:fig_6} and \ref{fig:fig_8} depict 
the capacity performance 
as a function of $\SNR_{\SU} \eqdef \PSU/\sigma^{2}$ for
different values of the ratio $d_{12}/d_{13}$, with $\PSU/\PPU = 1$.
Results show that, regardless of the availability of CSI at the STx, the performance of the proposed schemes 
(with and without VCs) rapidly improves when the STx
is moving away from the PTx and, at the same time, it is approaching the SRx.
This is a consequence of the fact that, when the distance $d_{12}$ between 
the PTx and the STx tends to be smaller than the distance $d_{24}$ between
the STx and the SRx, the signal-to-interference ratio at the SRx 
increases. The conventional OCR scheme is able to compete 
with the proposed scheme with VCs only
when the interference 
generated by the PU transmission over the $1 \to 4$ link
dominates the SU signal, i.e., the STx is too close to the PTx
and, at the same time, too far from the SRx.
In particular, when $d_{12}/d_{13} =0.3$, the conventional OCR scheme slightly outperforms 
the proposed scheme with VCs if there is no CSI at the STx (see Fig.~\ref{fig:fig_8}).
The motivation is that the uniform power allocation is suboptimal 
(in the information-theoretic sense) for the SU  when many
subchannels are heavily contaminated by the PU interference. 
This problem is circumvented if CSI is available at the STx and, hence, 
power can be optimally allocated on the subcarriers
(see Fig.~\ref{fig:fig_6}). 
Underneath all of this,
it is noteworthy that the OCR scheme necessarily requires the presence of
VCs (i.e., spectrum holes) in the PU signal, whose presence might be difficult to 
reliably detect in practice. In contrast, our scheme can achieve satisfactory data 
rates even without exploiting the (possible) presence of VCs, especially
when the STx is sufficiently close to the SRx. In particular, 
when $d_{12}/d_{14} > 0.3$,
the proposed scheme without VCs ensure a significant increase in data rate
with respect to the NOCR scheme proposed in \cite{Verde1,Verde2}
(see Fig.~\ref{fig:fig_8}), which not only is unable to exploit the presence of the VCs, but also assumes the transmission of one SU symbol per PU data block,
i.e, $N=1$, and does not carry out the precoding of the SU data.

Figs.~\ref{fig:fig_7} and \ref{fig:fig_9} report 
the capacity performance 
as a function of $\SNR_{\SU}$ for
different values of the power ratio $\PSU/\PPU$, with 
$d_{12}/d_{14}=0.7$. Overall, it is evident that,
compared to the other considered schemes, 
the performance advantage offered by the proposed
scheme (with and without VCs) becomes more and
more marked when $\PSU/\PPU$ increases. We remember
that  an increase in the power
ratio $\PSU/\PPU$ is also beneficial for the PU system (see Fig.~\ref{fig:fig_4}).

For example, with reference to a primary Wi-Fi system with 
$1/T_c=20$ MHz, when $\PPU=\PSU$, $d_{12}/d_{14}=0.7$, 
and $\SNR_{\SU} = 20$ dB, it results
from Figs.~\ref{fig:fig_6} and \ref{fig:fig_8}
that the SU capacity of the proposed scheme with VCs
is at least equal to $11$ Mbps with
CSIT and $8$ Mbps with no CSIT, 
whereas, when $\PSU$ is 
twice $\PPU$, these gains go up at least to $14$ Mbps
and $10$ Mbps (see Figs.~\ref{fig:fig_7} and \ref{fig:fig_9}), respectively.

\section{Conclusions}
\label{sec:concl}

We proposed a spectrum sharing scheme which allows the SU to concurrently transmit within the overall bandwidth of the PU system, by  
generalizing and subsuming as a particular case existing
OCR and NOCR approaches. 
Contrary to the classical NOCR paradigm, 
a key feature of the proposed scheme is that the concurrent SU
transmission improves (rather than degrades)
the performance of the PU system, under reasonable conditions.
Another remarkable result is that, if the SU is willing to spend extra transmit power, it can obtain a multicarrier link with a significant data rate.
Such a performance might be further improved by assuming that
the STx has perfect knowledge of the relevant channel parameters,
which allows one to use the waterfilling solution for precoding its 
information-bearing data. 

%%%%%%%%%Appendices%%%%%%%%%%%%%%%%%%%%%

\appendices

\section{Proof of Lemma~\ref{lem:cpu}}
\label{app:lemma-cpu}
Since $\|\bm{\alpha}_m\|^2$ in \eqref{eq:a} is a strictly increasing function
of $\PSU$, $\forall m \in \mathcal{M}$. 
it is sufficient to show that
the first-order partial derivative of $\CPUlower$ with respect to 
$\|\bm{\alpha}_m\|^2$
is non-negative, $\forall m \in \mathcal{M}$.
Starting from \eqref{eq:EE-capacity}, one has
\be
\frac{\partial}{\partial \|\bm{\alpha}_m\|^2} \, \CPUlower =
\frac{\log_2(e)}{M} \sum_{m \in \IPUuc}
\frac{\partial}{\partial \|\bm{\alpha}_m\|^2} \, \expect[\Psi(\ASNRP)] \: .
\ee
At this point, let $X_m \eqdef (\bm{\alpha}^\herm_m \, \xsuuc)/\|\bm{\alpha}_m\|$,
we can equivalently rewrite \eqref{eq:gamma3m} as follows
$\ASNRP =
\ASNRd \,
\left({1 + Z_m \, \|\bm{\alpha}_m\|^2  \, \frac{\sigma_{12}^2}{\sigma_{13}^{2}}}\right)/
\left({1+ \|\bm{\alpha}_m\|^2 \, \sigma_{23}^2  \, \frac{\sigma_{v_2}^2}{\sigma_{v_3}^{2}}}\right)$,
where the random variable $Z_m \eqdef |H_{23}(m)|^2 \cdot |X_m|^2$
is the product of two independent exponential random variables
with mean $\sigma_{23}^2$ and $1$, respectively, whose probability
density function is denoted by $f(z)$. It can be seen \cite{Papoulis} that
$f(z) \equiv 0$ for $z <0$, whereas
\be
f(z) = \frac{1}{\sigma_{23}^2}
\int_{0}^{+ \infty} \frac{1}{x} \, e^{-\left(\frac{x}{\sigma_{23}^2}+
\frac{z}{x}\right)} \, {\rm d}x =
\frac{2}{\sigma_{23}^2} \,
K_0\left(\frac{\sqrt{z}}{\sigma_{23}}\right) \: ,
\quad
\text{for $z \ge 0$}
\ee
where we have also used the result reported in footnote~\ref{foot:bessel}.
It results that
\be
\frac{\partial}{\partial \|\bm{\alpha}_m\|^2} \, \expect[\Psi(\ASNRP)] =
\frac{\partial}{\partial \|\bm{\alpha}_m\|^2} \int_{0}^{+\infty}
\left[ \int_{0}^{+\infty} e^{-u} \,  \psi(\|\bm{\alpha}_m\|^2,z,u)  \, \mathrm{d}u \right] f(z) \, {\rm d}z
\label{lem:dev-int}
\ee
where 
$\psi(\|\bm{\alpha}_m\|^2,z,u) \eqdef \ln\left[1+
\ASNRd \,
\left({1 + z \, \|\bm{\alpha}_m\|^2 \, \frac{\sigma_{12}^2}{\sigma_{13}^{2}}} \, u\right)
/\left({1+ \|\bm{\alpha}_m\|^2 \, \sigma_{23}^2  \, \frac{\sigma_{v_2}^2}{\sigma_{v_3}^{2}}}\right) \right]$.
As a consequence of the Lebesgue's dominated convergence (see, e.g., \cite{Rudin_book}), we can interchange the order of differentiation and double integration in \eqref{lem:dev-int} because it holds that

\begin{enumerate}[(i)]

\itemsep=0mm

\item
$\psi(a,z,u)$ is differentiable for any $a>0$, $z \ge 0$, and $u \ge 0$, and it results that
\be
\frac{\partial}{\partial a} \psi(a,z,u) =
\left(1+
\ASNRd \,
\frac{1 + z \, a \, \frac{\sigma_{12}^2}{\sigma_{13}^{2}}}
{1+ a \, \sigma_{23}^2  \, \frac{\sigma_{v_2}^2}{\sigma_{v_3}^{2}}}\, u\right)^{-1}
\ASNRd \, \frac{z \, \frac{\sigma_{12}^2}{\sigma_{13}^{2}}-\sigma_{23}^2  \, \frac{\sigma_{v_2}^2}{\sigma_{v_3}^{2}}}{\left(1+ a \, \sigma_{23}^2  \, \frac{\sigma_{v_2}^2}{\sigma_{v_3}^{2}}\right)^2} \, u \: ;
\ee

\item
$e^{-u} \, \frac{\partial}{\partial a}  \psi(a,z,u)$ is summable
with respect to $u$ in $(0, + \infty)$ $\forall a>0$ e $z \ge 0$;

\item
the function
$f(z) \int_{0}^{+\infty} e^{-u} \,  \frac{\partial}{\partial a} \psi(a,z,u)  \, \mathrm{d}u$
is summable with respect to $z$ in $[0,\infty)$ $\forall a >0$.

\end{enumerate}
Since $z \ge 0$, $u \ge 0$, $f(z) \ge 0$, and
$e^{-u}>0$, it is readily proven that
${\partial}/{\partial \|\bm{\alpha}_m\|^2} \, \CPUlower \ge 0$ and, hence,
$\CPUlower$ is a  monotonically increasing function of $\PSU$, if
\be
z \, \frac{\sigma_{12}^2}{\sigma_{13}^{2}}-\sigma_{23}^2  \, \frac{\sigma_{v_2}^2}{\sigma_{v_3}^{2}} \ge 0
\, \Longleftrightarrow \,
z \ge \sigma_{23}^2  \, \frac{\sigma_{13}^2}{\sigma_{12}^{2}} \, \frac{\sigma_{v_2}^2}{\sigma_{v_3}^{2}} \: .
\label{lem:cond}
\ee
Since $\sigma_{23}^2$ is bounded when
$({\sigma_{13}}/{\sigma_{12}}) \,  ({\sigma_{v_2}}{\sigma_{v_3}})
\rightarrow 0$,\sfootnote{For instance, if $\sigma_{v_2}^2 \approx \sigma_{v_3}^2$ and assuming the path-loss model $\sigma_{i \ell}^2 = d_{i \ell}^{-\eta}$, it comes from Carnot's cosine law
that $\sigma_{23}^2 \rightarrow \sigma_{13}^2$.}
when condition \eqref{eq:cond-2} holds, inequality \eqref{lem:cond}
is trivially satisfied because it ends up to $z \ge 0$.

%%%%%%%%%%%%ACK%%%%%%%%%%%%%%%%%%%%%%%%

\section*{Acknowledgement}

The authors would like to thank Prof.~Anna Scaglione for the discussion
regarding  the latency issues of the proposed scheme, which helped us 
improve an earlier version of the manuscript. 

%%%%%%%%%%%%%%%%%% Bibliography%%%%%%%%%%%%%%%

\end{document}